\documentclass[format=acmsmall, nonacm=true, screen=true]{acmart}
\def\BibTeX{{\rm B\kern-.05em{\sc i\kern-.025em b}\kern-.08emT\kern-.1667em\lower.7ex\hbox{E}\kern-.125emX}}

\usepackage{subcaption}

\usepackage{listings}
\usepackage[ruled]{algorithm2e}

\SetAlFnt{\small}
\SetAlCapFnt{\small}
\SetAlCapNameFnt{\small}
\SetAlCapHSkip{0pt}
\IncMargin{-\parindent}

\DeclareMathOperator{\Link}{Lk}
\DeclareMathOperator{\Star}{St}
\DeclareMathOperator{\Closure}{Cl}
\DeclareMathOperator{\Orient}{Orient}
\DeclareMathOperator{\Level}{Lvl}
\DeclareMathOperator{\Tangent}{Tangent}

\newcommand{\st}{\ \big| \ }
\newcommand{\mc}[1]{\mathcal{#1}}
\newcommand{\abs}[1]{ {\lvert {#1} \lvert} }
\newcommand{\largewedge}{\mbox{\Large $\wedge$}}

\newcommand{\asc}{CASC\xspace}
\newcommand{\gamer}{\texttt{GAMer}\xspace}
\newcommand{\simplex}[1]{$#1$-simplex}
\newcommand{\simplices}[1]{$#1$-simplices}

\begin{document}
\title[Colored Abstract Simplicial Complex]{The Implementation of the Colored Abstract Simplicial Complex and its Application to Mesh Generation}

\author{Christopher T. Lee}
\authornote{Contributed equally to this work}
\orcid{0000-0002-0670-2308}
\email{ctlee@ucsd.edu}

\author{John B. Moody}
\authornotemark[1]

\author{Rommie E. Amaro}
\orcid{0000-0002-9275-9553}
\email{ramaro@ucsd.edu}

\author{J. Andrew McCammon}
\orcid{0000-0003-3065-1456}
\email{jmccammon@ucsd.edu}

\author{Michael J. Holst}
\orcid{}
\email{mholst@ucsd.edu}
\affiliation{%
  \institution{University of California San Diego}
  \city{La Jolla}
  \state{CA}
  \postcode{92037}
  \country{USA}
}
\authorsaddresses{%
Author's addresses: C. T. Lee, R. E. Amaro, and J. A. McCammon, Department of Chemistry and Biochemistry, University of California San Diego; M. J. Holst, Department of Mathematics, University of California San Diego; J. B. Moody (Current address) ViaSat, Inc. Carlsbad-Bldg 10-2063, 6155 El Camino Real, Carlsbad, CA 92009
}

\begin{abstract}
  We introduce \asc: a new, modern, and header-only C++ library which provides a data structure to represent arbitrary dimension abstract simplicial complexes (ASC) with user-defined classes stored directly on the simplices at each dimension.
  This is accomplished by using the latest C++ language features including variadic template parameters introduced in C++11 and automatic function return type deduction from C++14.
  Effectively \asc decouples the representation of the topology from the interactions of user data.
  We present the innovations and design principles of the data structure and related algorithms.
  This includes a metadata aware decimation algorithm which is general for collapsing simplices of any dimension.
  We also present an example application of this library to represent an orientable surface mesh.
\end{abstract}

\keywords{Abstract Simplicial Complexes, Molecular Modeling, Mesh Generation, Mesh Decimation, Variadic Templates, C++ Library}

\maketitle

\renewcommand{\shortauthors}{Lee and Moody, et al.}

\section{Introduction}\label{sec:intro}

  \par For problems in computational topology and geometry, it is often beneficial to use simple building blocks to represent complicated shapes.
  A popular block is the simplex, or the generalization of a triangle in any dimension.
  Due to the ease of manipulation and the coplanar property of triangles, triangulations have become commonplace in fields such as geometric modeling and visualization as well as topological analysis.
  Discretizations are also used for efficient solving of Partial Differential Equations (PDE).
  The use of meshes has become increasingly popular even in the fields of computational biology and medicine\cite{Zhang2016}.

  \par As methods in structural biology improve and new datasets become available, there is interest in integrating experimental and structural data to build new predictive computer models\cite{Roberts2014}.
  A key barrier that modelers face is the generation of multi-scale, computable, geometric models from noisy datasets such as those from Electron Tomography (ET)\cite{Yu2008a}.
  This is typically achieved in at least two steps: (1) segmentation of relevant features, and (2) approximation of the geometry using meshes.
  Subsequently, numerical techniques such as Finite Elements Modeling or Monte Carlo can be used to investigate the transport and localization of molecules of interest.

  \par While many have studied mesh generation in the fields of engineering and animation, few methods are suitable for biological datasets.
  This is largely due to noise introduced by limits in image resolution or contrast.
  Even while using state-of-the-art segmentation algorithms for ET datasets there are often unresolved or missed features.
  Due to these issues, the generated meshes often have holes and other non-manifolds which must be resolved prior to mathematical modeling.
  Another challenge is the interpretation of a voxel valued segmentation.
  The conversion of zig-zag boundaries into a mesh can lead to other problems such as extremely high aspect ratio triangles or, in general, poorly conditioned elements\cite{Yu2008a}.
  To remedy this, various smoothing and decimation algorithms must also be applied prior to simulation.

  \par Previous work by us and others have introduced a meshing tool for biological models, \gamer, for building 3D tetrahedral meshes which obey internal and external constraints, such as matching embedding and/or enclosing molecular surfaces.
  It also provides the ability to use various mesh improvement algorithms for volume and surface meshes\cite{Gao2013,Yu2008}.
  \gamer uses the \texttt{Tetgen} library as the primary tetrahedral volume generator\cite{Si2015a}.
  While the algorithms are sound, the specific implementation is prone to segmentation faults even for simple meshes.
  Careful analysis of the code has identified that the data structures used for the representation of the mesh is primarily at fault.
  This article will focus entirely on the representation of topology in very complex mesh generation codes.
  We note that the algorithms which handle geometric issues like shape regularity and local adaptivity are well understood\cite{Babuska1976,Liu1994}, among others.
  Similarly there is a large body of literature related to local mesh refinement and decimation\cite{Bank1983,Bank1996}.
  Our innovations serve to enable the implementation of these algorithms in the most general and robust way.

  \par \gamer currently employs a neighbor list data structure which tracks the adjacency and orientation of simplices.
  Neighbor lists are quick to construct, however the representation of non-manifolds often leads to code instability.
  Algorithms must check for aberrant cases creating substantial overhead.
  We note that while the need to gracefully represent 2D and 3D non-manifolds for ET applications drove our initial focus, we are also interested in mesh generation in higher dimensions with applications to:
  numerical general relativity (3D+1)\cite{HSTV16a,Regge1961}, computational geometric analysis (nD)\cite{HoTi14a}, phase space simulations (6D), and arbitrary collective variable spaces in molecular modeling for enhanced sampling\cite{Vanden-Eijnden2009d}.
  We therefore chose following requirements for a mesh data structure to serve as design goals:
  \begin{itemize}
    \item General and capable of representing non-manifold, mixed dimensional, oriented and non-oriented meshes in arbitrary dimensions.
    \item Support for inline and flexible data storage. In some applications, data must be associated with the topology. For example, problems in general relativity typically require the storage of metric tensors on all simplex dimensions.
    \item Support for intuitive and simple manipulations and traversals.
  \end{itemize}

  \par Here we describe the development of a scalable \textit{colored abstract simplicial complex} data structure called \asc.
  Simplices are stored as nodes on a Hasse diagram.
  For ease of traversal all adjacency is stored at the node level.
  An additional data object can be stored at each node which is typed according to the simplex dimension at compile time.
  This means that, for example, for a mesh the \simplices{0} can be assigned a vertex type while the \simplices{2} can store some material property instead.
  Typing of each \simplex{k} is achieved using variadic templates introduced in C++11.
  \asc thus provides a natural separation between the combinatorics represented by the ASC from the underlying data types at each simplex dimension and their interactions.
  In \S\ref{sec:background} we briefly define an ASC and some relevant definitions followed by the introduction of the \asc data structure and it's construction in \S\ref{sec:casc}.
  We then demonstrate the use of \asc to represent a surface mesh and compute vertex tangents in \S\ref{sec:applications}.

  \subsection{Related Work}

    \par Although many data structures to represent simplicial complexes have been developed, to the best of our knowledge there currently exists no data structure which supports meshes of arbitrary dimension with user-selected typed data stored directly on each simplex.
    A full review of all existing data structures is beyond the scope of this work,
    however we highlight several representative examples.
    Many data structures such as the half-edge and doubly-connected edge list among others are restricted to the representation of two-manifolds only\cite{DeFloriani2005}.
    Other data structures such as SIG\cite{DeFloriani2004}, IS\cite{DeFloriani2010b}, IA*\cite{Canino2011}, SimplexTree\cite{Boissonnat2012}, AHF\cite{Dyedov2015,Ray2015}, LinearCellComplex, and dD Triangulations\cite{boissonnat:inria-00412437} support or can be extended to represent arbitrary dimensional simplicial complexes.
    However, their current implementations either do not consider the storage of data beyond possibly embedding, or do not support inline storage of user data.
    LinearCellComplex from CGAL supports only a linear geometrical embedding\cite{CGAL}.
    AHF implemented in MOAB uses separate arrays of data which are then referenced using a handle\cite{Dyedov2015,Ray2015}.
    In addition to the limitations of data storage, some make assumptions limiting their generality.
    dD Triangulations, for example, assumes that a simplicial complex is pure and therefore does not support the representation of mixed dimensional complexes\cite{boissonnat:inria-00412437}.

\section{Background -- Abstract Simplicial Complexes}\label{sec:background}
  \par An Abstract Simplicial Complex (ASC) is a combinatorial structure which can be used to represent the connectivity of a simplicial mesh, independent of any geometric information. More formally, the definition of an ASC is as follows.

  \begin{definition}
  Given a vertex set $V$, an \textit{abstract simplicial complex} $\mc{F}$ of $V$ is a set of subsets of $V$ with the following property: for every set $X \in \mc{F}$, every subset $Y \subset X$ is also a member of $\mc{F}$.
  \label{def:asc}
  \end{definition}

  \par The sets $s \in \mc{F}$ are called a simplex or face of $\mc{F}$; similarly a face $X$ is said to be a face of simplex $s$ if $X\subset s$. Since $X$ is a face of $s$, $s$ is a coface of $X$. Each simplex has a dimension characterized by $\dim{s} = \abs{s}-1$, where $\abs{s}$ is the cardinality of set $s$. A simplex of $\dim{s} = k$ is also called a $k$-simplex. The dimension of the complex, $\dim(\mc{F})$, is defined by the largest dimension of any member face. Simplices of the largest dimension, $\dim(\mc{F})$ are referred to as the facets of the complex.

  \par If one simplex is a face of another, they are incident. Every face of a $k$-simplex $s$ with dimension $(k-1)$ is called a boundary face while each incident face with dimension $(k+1)$ is a coboundary face. Two $k$-simplices, $f$ and $s$ are considered adjacent if they share a common boundary face, or coboundary face. The boundary of simplex $s$, $\partial s$, is the sum of the boundary faces.

  \par Having introduced the concept of an ASC, we can also define several operations useful when dealing with ASCs. A subcomplex is a subset that is a simplicial complex itself. The Closure ($\Closure$) of a simplex, $f$, or some set of simplices $F \subseteq \mc{F}$ is the smallest simplicial subcomplex of $\mc{F}$ that contains $F$:
  \begin{equation}
    \Closure(f) = \{s \in \mc{F} \st s \subseteq f\};\qquad \Closure(F) = \bigcup\limits_{f\in F} \Closure(f) \quad\mathrm{(closure)}.
    \label{eq:closure}
  \end{equation}
  It is often useful to consider the local neighborhood of a simplex. The Star ($\Star$) of a simplex $f$ is the set of all simplices that contain $f$:
  \begin{equation}
    \Star(f) = \{s \in \mc{F} \st f \subseteq s\}; \qquad \Star(F) = \bigcup\limits_{f \in F} \Star(f) \quad\mathrm{(star)}.
    \label{eq:star}
  \end{equation}
  The Link ($\Link$) of $f$ consists of all faces of simplices in the closed star of $f$ that do not intersect $f$:
  \begin{equation}
    \Link(f) = \{s \in \Closure \circ \Star(f) \st s\cap f = \emptyset \} = \Closure\circ\Star(f) - \Star\circ\Closure(f) \quad\mathrm{(link)}.
    \label{eq:link}
  \end{equation}

  For some algorithms, it is often useful to iterate over the set of all vertices or edges etc. We use the following notation for the horizontal ``level'' of an abstract simplicial complex.
  \begin{equation}
    \Level_k(\mathcal{F}) = \left\{ s \in \mc{F} \st \dim{s} = k \right\}
  \end{equation}
  A subcomplex which contains all simplices $s \in \mc{F}$ where $\dim(s) \leq k$ is the $k$-skeleton of $\mc{F}$:
  \begin{equation}
    \mc{F}_k = \Closure\circ\Level_k(\mc{F}) = \bigcup\limits_{i\leq k} \Level_i(\mc{F}).
  \end{equation}

  \begin{figure}[ht!]
    \centering
    \includegraphics[width=\textwidth]{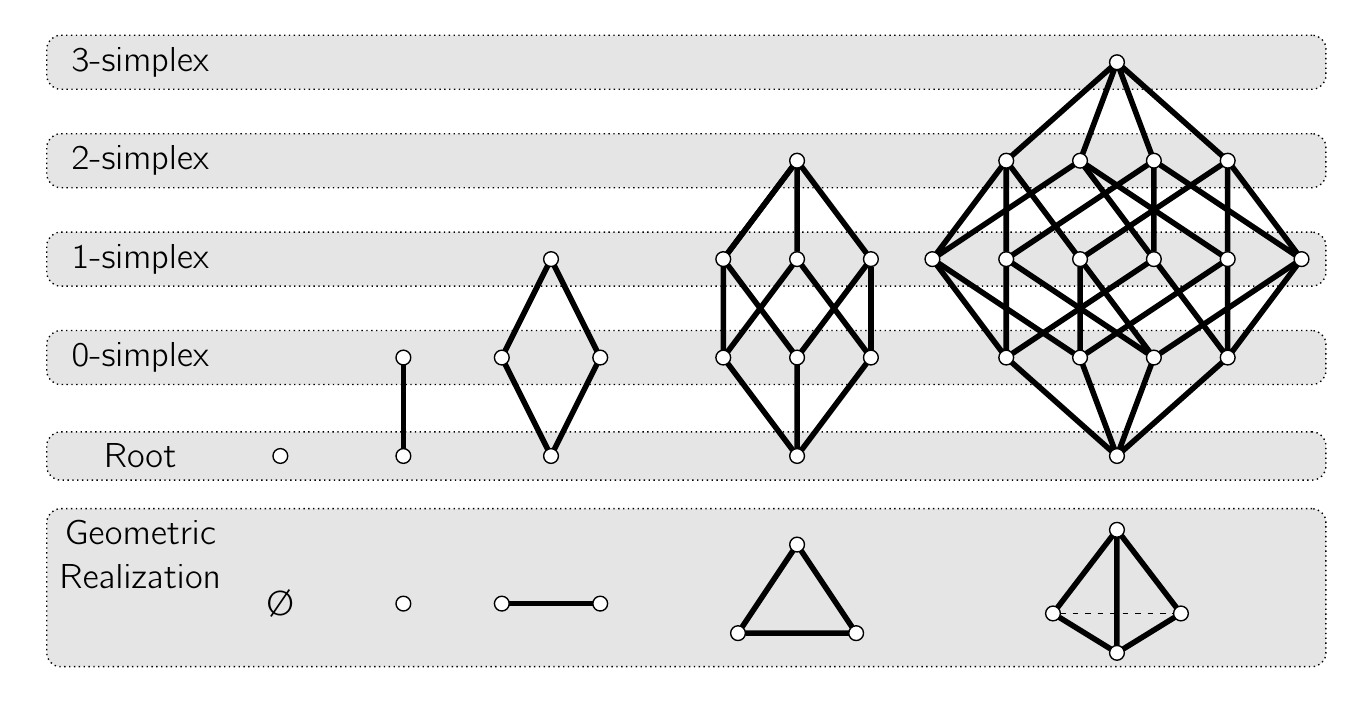}
    \caption{Hasse diagrams of several Abstract Simplicial Complexes and a geometric realization from left to right: the empty set, a vertex, an edge, a triangle, a tetrahedron.}
    \label{fig:hasse}
  \end{figure}

  \par By Definition~\ref{def:asc}, an ASC forms a partially ordered set, or poset.
  Posets are frequently represented by a Hasse diagram, a directed acyclic graph, where nodes represent sets, and edges denote set membership.
  Several example simplicial complexes and their corresponding Hasse diagrams are shown in Fig.~\ref{fig:hasse}.
  Colloquially we will use \textit{up} and \textit{down} to refer to the boundary and coboundary of a simplex respectively.
  In Hasse diagrams, we follow a convention that simplices shown graphically on the same horizontal level have the same simplex dimension.
  Furthermore, simplices of greater dimension are drawn above lesser simplices.

  \begin{figure}[ht!]
    \centering
    \includegraphics[width=\textwidth]{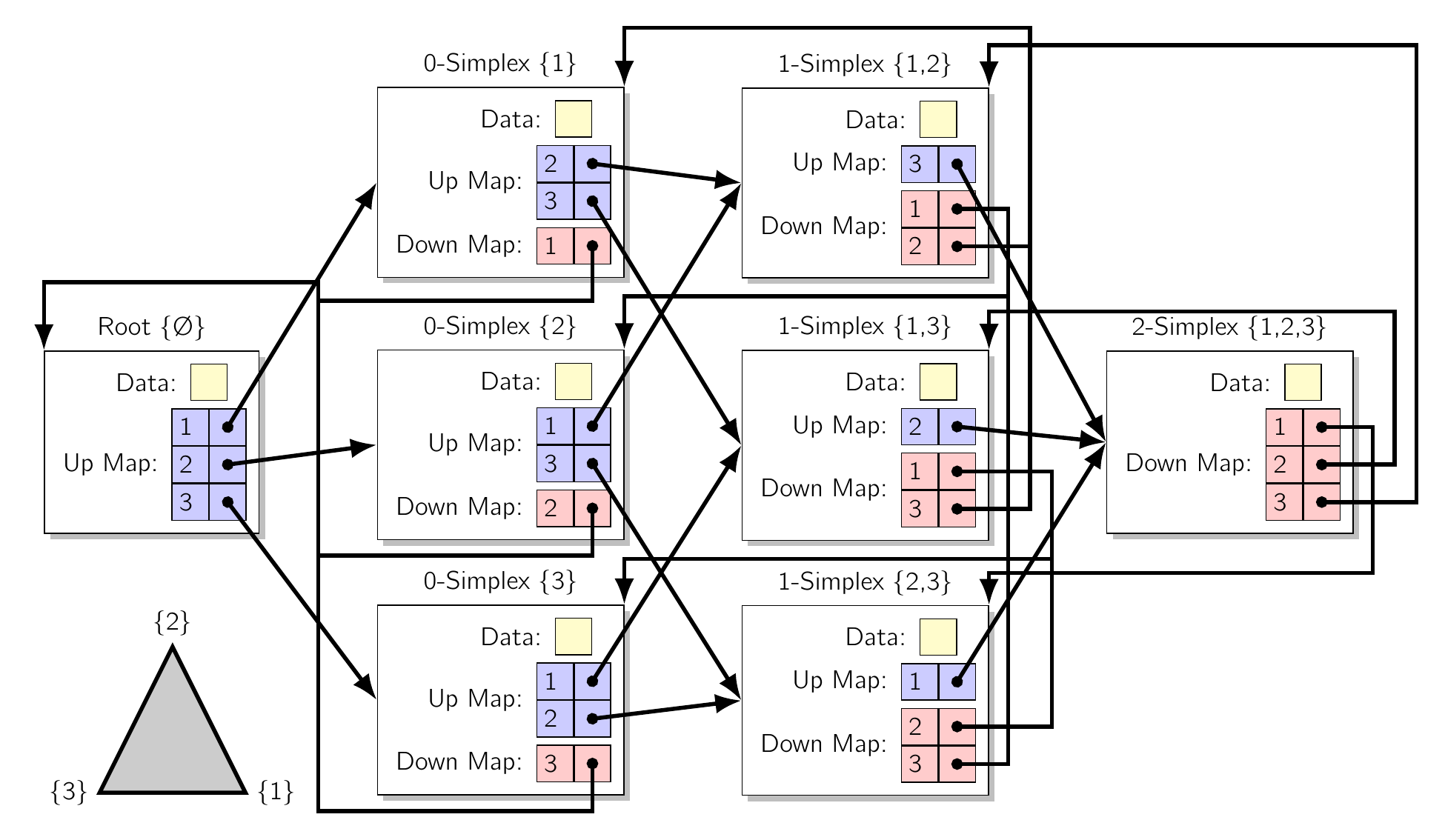}
    \caption{Data structure diagram of a triangle represented using \asc. Each simplex is represented as a node containing a dictionary up and/or down which maps the vertex index to a pointer to the next simplex. Data can be stored at each node with type determined at compile time. Effectively each level can contain different metadata as defined by the user, separating the interactions of user data from the representation of topology.}
    \label{fig:datastruct}
  \end{figure}

\section{Colored Abstract Simplicial Complex}\label{sec:casc}

  \par In this section we introduce the \asc data structure and its implementation.
  For a given simplicial complex, each simplex is represented by a node (\verb|asc_Node|) in the Hasse diagram, and defined by a set of keys corresponding to the vertices which comprise the simplex.
  Note that we use node to refer to objects in \asc Hasse diagram and not \simplices{0}. Instead, \simplex{0} are referred to as the vertices of the mesh.
  Furthermore we refer to the \simplex{\varnothing} or \simplex{-1} as the root simplex interchangeably.
  When a node is instantiated, we assign it a unique Integer Internal Identifier (iID) for use in the development of \asc algorithms.
  The iID is constant and never exposed to the end-user except for debugging purposes.
  Instead nodes can be referenced by the user using the \verb|SimplexID| which acts as a convenience wrapper around an \verb|asc_Node*|, providing additional support for move semantics for fast data access.
  All topological relations (i.e., edges of the Hasse diagram) are stored in each node as a dictionary which maps user specified keys to \verb|SimplexID|s up and down.
  An example data structure diagram of triangle \{1,2,3\} is shown in Fig.~\ref{fig:datastruct}.
  Based upon this example, if a user has the \verb|SimplexID| of \simplex{1} $\{1,2\}$ and wishes to get \simplex{2} $\{1,2,3\}$, they can look in the \verb|Up| dictionary of \verb|SimplexID|$\{1,2\}$ for key $3$ which maps to a \verb|SimplexID|$\{1,2,3\}$.
  The vertices which constitute each simplex are not stored directly, but can be accessed by aggregating all keys in \verb|Down|.

  \par We note that while the representation of all topological relations is redundant and may not be memory optimal, it vastly simplifies the traversals across the complex.
  Furthermore, the associate algorithms and innovations using variable typing are general and thus compatible with other more condensed representations.

  \subsection{Variable Typing Per Simplex Dimension}

    \par We achieve coloring by allowing user-defined data to be stored at each node.
    The typical challenge for strongly typed languages such as C++ is that the types must be defined at compile time.
    Typical implementations would either hard code the type to be stored at each level or use a runtime generic type, such as \verb|void*|.
    However, each of these have drawbacks. For the former, this requires writing a new node data structure for every simplicial complex we may wish to represent.
    For the latter, using \verb|void*| adds an extra pointer dereference which defeats cache locality and may lead to code instability.
    Another possible implementation might be to require users of the library to derive their data types from a common class through inheritance.
    This solution puts an unnecessary burden on users who may have preexisting class libraries, or simply wish to store a built in type, such as an \verb|int|.
    Furthermore, under the inheritance scheme, changes to the underlying container may require users to update their derived classes.
    To avoid this cumbersome step, we have employed the use of variadic templates introduced in C++11 to allow for unpacking and assignment of data types.
    The user specifies the types to be stored at each level in a list of templates to the object constructor, see Fig.~\ref{fig:unpacking}.

    \begin{figure}[ht!]
      \centering
      \includegraphics[width=\textwidth]{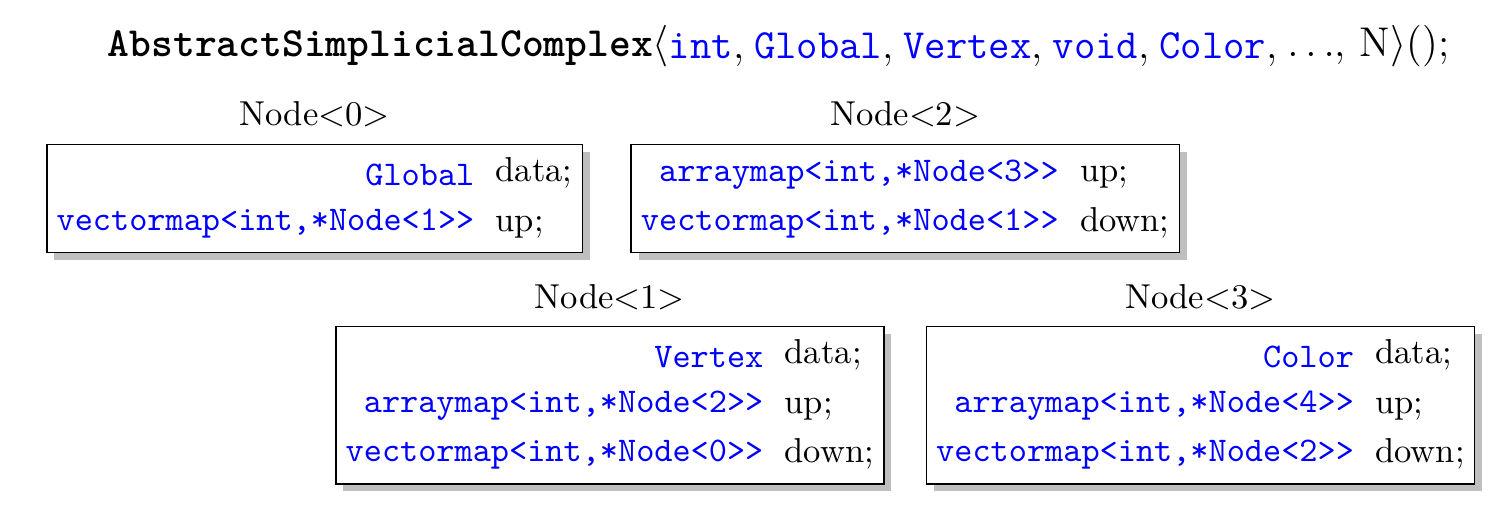}
      \caption{
        Template arguments for ASC are unpacked and assigned to nodes accordingly. The first argument ``\texttt{int}'' is the key type for labeling vertices while the following arguments define node data types. Notably, \texttt{Node<2>} does not allocate memory for data as the corresponding template argument was ``\texttt{void}''. By passing in additional types, simplicial complexes of higher dimensions are instantiated.
      }
      \label{fig:unpacking}
    \end{figure}

    \par The variadic templating allows \asc to represent complexes of any user-defined dimension.
    To specify an $N$-simplicial complex, \asc requires the definition of an index/key type followed by $N+1$ data types and $N$ edge types. The first data type provided after the key type corresponds to data stored on the \simplex{\varnothing} which can be thought of as global metadata.
    For example, suppose we have a 2-simplicial complex intended for visualization and wish to store locations of vertices and colors of faces.
    A suitable ASC can be constructed using the following template command:
\begin{lstlisting}
auto mesh = AbstractSimplicialComplex<int, void, Vertex, void, Color>()
\end{lstlisting}
    If we now wish to represent a tetrahedral mesh, instead of constructing a new data structure, we can simply adjust the command:
\begin{lstlisting}
auto mesh = AbstractSimplicialComplex<int, void, Vertex, void, Color, Density>()
\end{lstlisting}
    In both cases, the first template argument is the key type for referring to vertices followed by the data type for each \simplex{k}. Supposing that the user does not wish to store data on any given level, by passing ``\texttt{void}'' as the template argument, the compiler will optimize the node data type and no memory will be allocated to store data. In both cases, the 0- and \simplices{1} will have no data.

    \par By using variadic templates, we allow the user to specify both the dimension of the simplicial complex as well as the types stored at each level. Because the type deduction is performed at compile time, there is no runtime performance impact on user codes. There is, however, some additional code complexity introduced. If the user wishes to retrieve the data stored in a simplicial complex, they must know what level they are accessing at compile time. A consequence is that the exposed identifier object, \verb|SimplexID|, is templated on the integral level, so that types can be deduced. This does not present a problem for simple use cases, such as:
\begin{lstlisting}
// Create alias for 2-simplicial complex
using ASC = AbstractSimplicialComplex<int,int,int,int>;
ASC mesh = ASC();  // construct the mesh object
mesh.insert<2>({1,2},5); // insert edge {1,2} with data 5
ASC::SimplexID<2> s = mesh.get_simplex<2>({1,2}); // get the edge
// NOTE: the type is templated on the integral level
std::cout << *s << std::endl; // prints data "5"
\end{lstlisting}
    However, when implementing algorithms intended to be generic on any simplicial complex, templated code must be written. We discuss the implementation of several such algorithms in the following section.

\section{Implemented Algorithms}\label{sec:algorithms}

  \par The following algorithms are provided with the \asc library:
  \begin{itemize}
    \item Basic Operations
    \begin{itemize}
      \item Creating and deleting simplices (\verb|insert|/\verb|remove|)
      \item Searching and traversing topological relations (\verb|GetSimplexUp|/\verb|GetSimplexDown|)
    \end{itemize}

    \item Traversals
    \begin{itemize}
      \item By level (\verb|get_level|)
      \item By adjacency (\verb|neighbors_up|/\verb|neighbors_down|)
      \item Traversals across multiple node types (Visitor Design Pattern/double dispatch)
    \end{itemize}

    \item Complex Operations
    \begin{itemize}
      \item Star/Closure/Link
      \item Metadata aware decimation
    \end{itemize}
  \end{itemize}

    \begin{algorithm}[ht!]
      \KwIn{$keys$[$n$]: Indices of $n$ simplices to describe new simplex $s$\\
      $rootSimplex$: The simplex to insert relative to (most commonly $root$)}
      \KwOut{The new simplex $s$}

      \SetStartEndCondition{ (}{)}{)}
      \AlgoDisplayBlockMarkers\AlgoDisplayGroupMarkers\SetAlgoBlockMarkers{ \{}{ \}\ }%
      \SetAlgoNoEnd\SetAlgoNoLine\DontPrintSemicolon
      \SetStartEndCondition{ (}{)}{)}
      \SetKwFor{For}{for}{}{}
      \SetKwIF{If}{ElseIf}{Else}{if}{}{elif}{else}{}

      \SetKwProg{Fn}{Function }{}{}
      \SetKwFunction{insert}{insert}
      \SetKwFunction{create}{createSimplex}

      \Fn(){\insert{keys[n], rootSimplex}}{
        \For(){i = 0; i < n; i++}{
          newSimplex = \create{$rootSimplex \cup keys[i]$}\;
          \eIf(\tcc*[h]{Recurse to insert sub-simplices}){i > 0}{
            \tcc*[l]{Pass only the first part of $keys$ up to index $i$}
            \Return \insert{keys[0:i], newSimplex}
          }
          (\tcc*[h]{Terminal conditional}){
            \Return newSimplex
          }
        }
      }

      \caption{Insertion of a new simplex}
      \label{alg:insertion}
    \end{algorithm}

    \begin{figure}[ht!]
      \centering
      \includegraphics[width=0.7\textwidth]{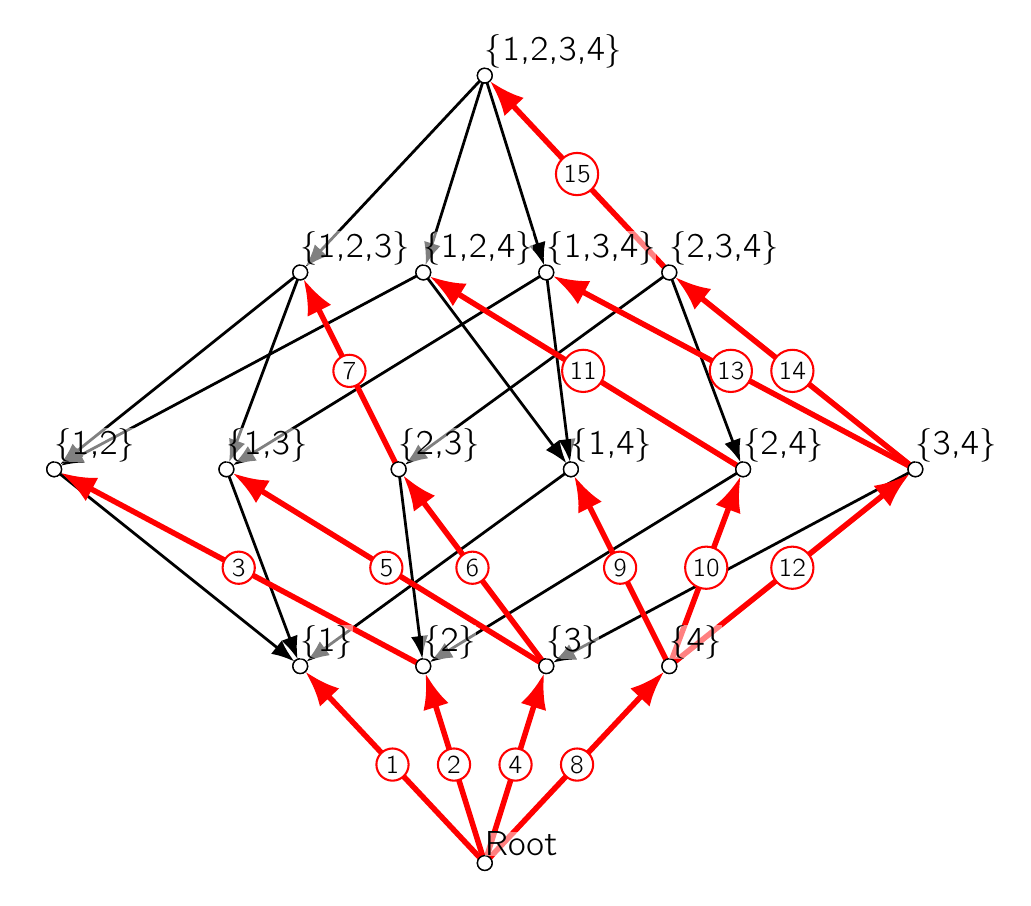}
      \caption{Recursive insertion of tetrahedral simplex \{1,2,3,4\}. The order of node insertions is represented by the numbered red arrows. When each node is created, the black arrows to parent simplices are created by backfilling.}
      \label{fig:insert}
    \end{figure}

  \subsection{Basic Operations}
    \subsubsection{Creating and deleting simplices}

      \par Since the \asc data structure maintains every simplex in the complex and all topological relations, inserting a \simplex{k}, $s$, into the complex means ensuring the existence of, and possibly creating, $\mathcal{O}(2^{k})$ nodes and $\mathcal{O}(k\cdot2^{k-1})$ edges (see derivation in SI).
      Fortunately, the combinatorial nature of simplicial complexes allows this to be performed recursively.
      A generalized recursive insertion operation for any dimensional complex and user specified types, is described in Algorithm~\ref{alg:insertion}.
      The insertion algorithm defines an insertion order such that all dependent simplices exist in the complex prior to the insertion of the next simplex.

      \par As an illustrative example of the template code used in this library, Algorithm~\ref{alg:insertion} is rewritten in C++ template function-like pseudocode shown in Algorithm~\ref{alg:templatedinsertion}. While the templated code is more complicated, it provides many optimizations. For example, since the looping and recursion are performed at compile time, for any \simplex{k} we wish to insert, any modern compiler should optimize the code into a series of \verb|insertNode()| calls; the \verb|setupForLoop()| and \verb|forLoop()| function calls can be completely eliminated. As a result, the optimized templated code will exhibit superior run time performance. To illustrate the insertion operation, a graphical representation of inserting tetrahedron \{1,2,3,4\} by Algorithm~\ref{alg:templatedinsertion} is shown in Fig.~\ref{fig:insert}, and step-by-step in Fig.~S\ref{fig:insertorder}. In the example, new simplex $root\cup v$ is added sequentially to the complex and any missing topological relations are found by traversing the faces of $root$ and backfilling.

      \par The removal of any simplex is also performed using a recursive template function.
      When removing simplex $s$, in order to maintain the property of being a simplicial complex, all cofaces of $s$ or $f \in \Star(s)$ must also be removed along with any boundary/coboundary relations of $f$.
      The implemented removal algorithm traverses up the complex and removes simplex $s$ and all cofaces of $s$ level by level.

    \begin{algorithm}[ht!]
      \KwIn{$root$: the simplex to start the search from.\\
      $keys$[$n$]: the relative name of the desired simplex $s$}
      \KwOut{The simplex $s$}

      \SetStartEndCondition{ (}{)}{)}
      \AlgoDisplayBlockMarkers\AlgoDisplayGroupMarkers\SetAlgoBlockMarkers{ \{}{ \}\ }%
      \SetAlgoNoEnd\SetAlgoNoLine\DontPrintSemicolon
      \SetStartEndCondition{ (}{)}{)}
      \SetKwFor{For}{for}{}{}
      \SetKwIF{If}{ElseIf}{Else}{if}{}{elif}{else}{}

      \SetKwProg{Fn}{Function }{}{}
      \SetKwFunction{gSU}{getSimplexUp}
      \SetKwFunction{gR}{getRecurse}

      \medskip
      \Fn(){\gSU{root, keys[n], index}}{
        \If(){index < n}{
          Find keys[index] in root.up \tcc*[h]{Look for the edge up}\;
          \eIf(){Edge up is found}{
            \Return \gSU{$root\cup keys[index]$, keys, index+1}
          }
          (\tcc*[h]{Edge keys[index] doesn't exist}){
            \Return null pointer
          }
        }
        \lElse(\tcc*[h]{Terminal condition}){
          \Return root
        }
      }

      \caption{Searching for a new simplex}
      \label{alg:search}
    \end{algorithm}

    \subsubsection{Searching and traversing topological relations}

      \par The algorithms for retrieving a simplex as well as for basic traversals from one simplex to another across the data structure are the same.
      Given a starting simplex, and an array of keys up, the new simplex can be found recursively by Algorithm~\ref{alg:search}; The annotated code used is shown in \S\ref{sec:actualneighbor}.
      Traversals from one simplex to another require a key lookup followed by a pointer dereference and therefore occur in approximately constant time ($\mathcal{O}(1)$).
      Since all topological relations are stored, the traversal order across the array of keys does not matter.
      The same algorithm can be applied going down in dimension.
      For the retrieval of an arbitrary simplex, we start the search up from the root node of the complex.

  \subsection{Traversals}

    \par Thus far, we have presented algorithms for the creation of a simplicial complex as well as the basic traversal across faces and cofaces.
    For many applications, other traversals, such as by adjacency, may be more useful.
    We present several built-in traversal algorithms as well as the visitor design pattern for complicated operations.

    \subsubsection{By level}
      \par It is often useful to have a traversal over all simplices of the same level. For example, iterating across all vertices to compute a center of mass. To support this in an efficient fashion, simplices of the same dimension are stored in a level specific map of iIDs to node pointers. Notably, the map for each level is instantiated with the correct user specified node type with respect to level at compile time. To achieve this, we again use variadic templates to generate a tuple of maps, where each tuple element corresponds to the map for a specific level's node type.

      \par Since \verb|asc_node|s are templated on the integral level, we can use a template type map to map an integral sequence to the node pointer type,
      \begin{equation*}
        tuple\langle 1,2,3,\ldots\rangle \xrightarrow{Node\langle k\rangle *} tuple\langle Node\langle 1\rangle *, Node\langle 2\rangle *, Node\langle 3\rangle *, \ldots\rangle,
      \end{equation*}
      producing a tuple of integrally typed node pointers. Subsequently, we can map again to generate a tuple of maps,
      \begin{equation*}
        tuple\langle Node\langle 1\rangle *, Node\langle 2\rangle *,\ldots\rangle \xrightarrow{map<int,T>}
        tuple\langle map\langle int, Node\langle 1\rangle *\rangle, map\langle int, Node\langle 2\rangle *\rangle,\ldots\rangle.
      \end{equation*}
      By using this variadic template mapping strategy we now have the correct typenames assigned. Any level of the tuple can be accessed by getting the integral level using functions in the C++ standard library. Variations of this mapping strategy are also used to construct the \verb|SimplexSet| and \verb|SimplexMap| structures below.

      \par For end users, the implementation details are entirely abstracted away. Continuing from the example above, iteration over all vertices of simplicial complex, \verb|mesh|, can be performed using the provided iterator adaptors as follows.
\begin{lstlisting}[caption={Example use of iterator adaptors for traversal across vertices of mesh.},captionpos=b]
// Deduces the type of nid = ASC::SimplexID<1>
for(auto nid : mesh.get_level_id<1>()){
  std::cout << nid << std::endl;
}
\end{lstlisting}
      The function \verb|get_level_id<k>()| retrieves level $k$ from the tuple and returns an iterable range across the corresponding map.

      \begin{algorithm}[ht!]
        \KwIn{$s$: The simplex to get the neighbors of.}
        \KwOut{List of neighbors}

        \SetStartEndCondition{ (}{)}{)}
        \AlgoDisplayBlockMarkers\AlgoDisplayGroupMarkers\SetAlgoBlockMarkers{ \{}{ \}\ }%
        \SetAlgoNoEnd\SetAlgoNoLine\DontPrintSemicolon
        \SetStartEndCondition{ (}{)}{)}
        \SetKwFor{For}{for}{}{}
        \SetKwIF{If}{ElseIf}{Else}{if}{}{elif}{else}{}

        \SetKwProg{Fn}{Function }{}{}
        \SetKwFunction{gNU}{getNeighborsUp}
        \medskip
        \Fn(\tcc*[h]{Get the neighbors}){\gNU{s}}{
          Create an empty list of $neighbors$\;
          \tcc{Follow all coboundary relations up from s.}
          \For{$a \in s.up$}{
            SimplexID id = $s\cup a$\;
            \For(\tcc*[h]{Go down from id}){$b \in id$}{
              \lIf(\tcc*[h]{Do not add self to neighbors}){$id\setminus b \neq s$}{
                Add $id\setminus b$ to $neighbors$
              }
            }
          }
          \Return neighbors
        }

        \caption{Get the neighbors of a simplex, $s$, by inspecting the faces of $s$}
        \label{alg:neighbor}
      \end{algorithm}

    \subsubsection{By adjacency}
      \par Many geometric algorithms operate on the local neighborhood of a given simplex. Unlike other data structures such as the halfedge, \asc does not store the notion of the next simplex. Instead, adjacency is identified by searching for simplices with shared faces or cofaces in the complex. The algorithm for finding neighbors with shared faces is shown in Algorithm~\ref{alg:neighbor}. We note that the set of simplices with shared faces may be different than the set of simplices with shared cofaces. Both adjacency definitions have been implemented and we leave it to the end user to select the relevant function. Once a neighbor list has been aggregated, it can be traversed using standard methods. While the additional adjacency lookup step is extra in comparison to other data structures, in many cases, the generation of neighbor lists need only be done once and cached. The trade off is that \asc offers facile manipulations of the topology without having to worry about reorganizing neighbor pointers.

    \subsubsection{Traversals over multiple node types.}
      \par When performing more complicated traversals, such as iterating over the star of a simplex, multiple node types may be encountered. In order to avoid typename comparison based branch statements, we have implemented visitor design pattern-based breadth first searches (BFS). The visitor design pattern refers to a double dispatch strategy where a traversal function takes a visitor functor which implements an overloaded \verb|visit()| function. At each node visited, the traversal function will call \verb|visit()| on the current node. Since the functor overloads \verb|visit()| per node type, the compiler can deduce which visit function to call. Example pseudocode is shown in Listing~\ref{lst:visitor}. This double dispatch strategy, eliminates the need for extensive runtime typename comparisons, and enables easy traversals over multiple node types. We provide breadth first traversals up and down the complex from a set of simplices. These visitor traversals are used extensively in the complex operations described below.
\begin{lstlisting}[label={lst:visitor},caption={Example pseudocode of double dispatch to traverse the complex while
scaling the mesh by 2 and coloring the faces green.},captionpos=b]
template <typename Complex>
struct Visitor{
  template <std::size_t k>
  using Simplex = typename Complex::template SimplexID<k>;
  // General template prototype
  template <std::size_t level>
  bool visit(Complex& F, Simplex<level> s){
    return true;
  }
  // Specialization for vertices
  bool visit(Complex& F, Simplex<1> s){
    s.position = s.position*2;
    return true;
  }
  // Specialization for faces
  bool visit(Complex& F, Simplex<3> s){
    s.color = 'green';
    return true;
  }
};

void BFS(ASC& mesh, Visitor&& v){
  // ... traversal code
  v.visit(mesh, currentSimplex);
  // NOTE: visit is overloaded and called based on function prototype.
}

void main(){
  // ... define simplicial complex traits for a surface mesh
  ASC mesh = ASC(); // construct the mesh object
  // ... insert some simplices etc.
  BFS(mesh, Visitor());
  }
}
\end{lstlisting}

  \subsection{Complex Operations}
    \subsubsection{Star/Closure/Link}
      \par The star, link, and closure can be computed using the visitor breadth first traversals to collect simplices. These operations typically produce a set of simplices spanning multiple simplex dimensions, and thus simplex typenames, which cannot be stored in a traditional C++ set. We have implemented a multi-set data structure called the \verb|SimplexSet|, which is effectively a tuple of typed sets corresponding to each level. The \verb|SimplexSet| is constructed using the same mapping strategy as the tuple of maps used for the iteration across levels. For convenience, we provide functions for typical set operations such as insertion, removal, search, union, intersection, and difference. Using a combination of the star and closure functions with \verb|SimplexSet| difference we can get the link by Eq.~\ref{eq:link}.

    \subsubsection{Metadata aware decimation}
      \begin{figure}[ht!]
        \centering
        \begin{subfigure}[ht!]{0.7\textwidth}
          \includegraphics[width=\textwidth]{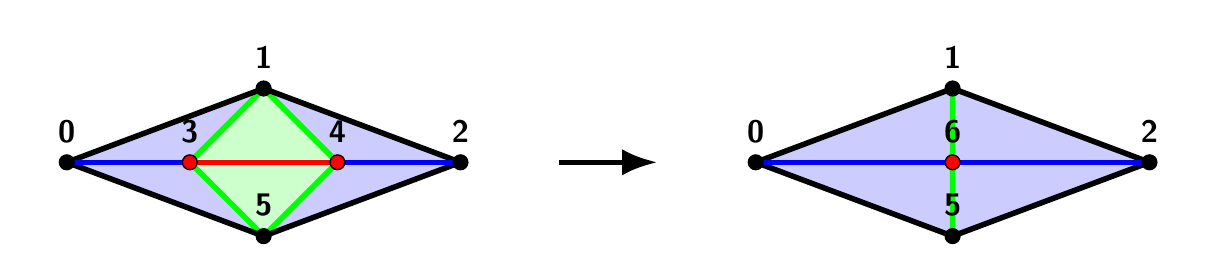}
          \caption{
            A geometric realization of the complex before and after decimation.
          }
          \label{fig:decexample}
        \end{subfigure}\\
        \begin{subfigure}[ht!]{0.7\textwidth}
          \includegraphics[width=\textwidth]{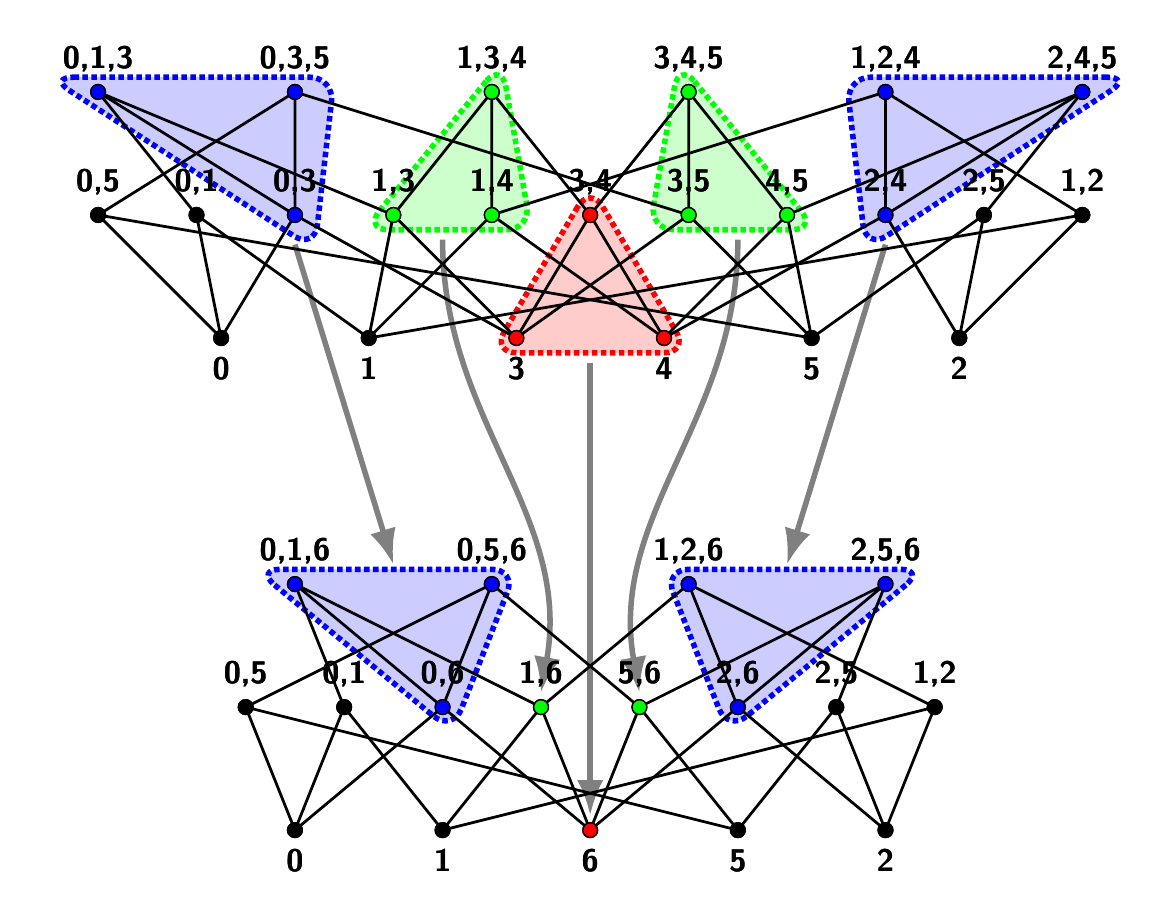}
          \caption{
            Explicitly drawn out Hasse diagrams for the constructed example where the TOP is before, and BOTTOM is after decimation. Grey arrows mark the relationships between sets of simplices before and after. Because there is always a mapping, users can define strategies to manage the stored data.
          }
          \label{fig:decdetail}
        \end{subfigure}
        \caption{The example decimation of edge $s = \{3,4\}$ in a constructed example.}
        \label{fig:decimation}
      \end{figure}

      \par We have implemented a general decimation algorithm which operates by collapsing simplices of higher dimensions into a vertex.
      While edge collapses for 2-manifolds are well studied, a general dimensional collapse is useful for decimating higher-dimensional meshes used to solve PDEs such as those encountered in general relativity.
      Since simplices are being removed from the complex, user data may be lost.
      Our implementation is metadata aware and allows the user to specify what data to keep post decimation.
      This is achieved by using a recursive algorithm to produce a map of removed simplices to new simplices.
      The user can use this mapping to define a function which maps the original stored data to the post decimation topology.
      This decimation strategy is implemented as an inplace operation yielding decimated mesh containing data mapped according to a user specified callback function.

      \begin{algorithm}[hp]
      \KwIn{$F$: the simplicial complex\;
      $s$: the simplex to decimate\;
      $clbk$: a callback function to handle the data mapping}
      \KwOut{Simplicial complex with simplex collapsed according to Def.~\ref{def:decimation}}

      \AlgoDisplayBlockMarkers\AlgoDisplayGroupMarkers\SetAlgoBlockMarkers{ \{}{ \}\ }%
      \SetAlgoNoEnd\SetAlgoNoLine\DontPrintSemicolon
      \SetStartEndCondition{ (}{)}{)}
      \SetKwFor{For}{for}{}{}
      \SetKwFor{BFSD}{BFSdown}{}{}
      \SetKwFor{BFSU}{BFSup}{}{}
      \SetKwIF{If}{ElseIf}{Else}{if}{}{elif}{else}{}

      \SetKwFunction{pR}{performRemoval}
      \SetKwFunction{pI}{performInsertion}
      \SetKwFunction{insert}{insert}
      \SetKwFunction{append}{append}
      \SetKwFunction{remove}{remove}
      \SetKwFunction{name}{name}

      \medskip
      Simplex np =  F.newVertex() \tcc*[h]{Create a dummy new vertex to map to}\;
      SimplexSet N \tcc*[h]{For the complete neighborhood}\;
      SimplexDataSet data \tcc*[h]{Data structure to store (simplex name, simplex data) pairs}\;
      SimplexMap simplexMap \tcc*[h]{Data structure to store New Simplex -> SimplexSet map}\;
      \medskip
      \tcc*[l]{Get the complete neighborhood (all simplices which are associated with $s$).}
      \For(){vertex $v$ of $s$}{
        \BFSU{simplex $i \in \Star(v)$}{
          \lIf{$j \notin N$}{N.\insert(j)}
      }}
      \medskip
      \tcc*[l]{Backup the complete neighborhood. These simplices will be destroyed eventually.}
      SimplexSet doomed = N\;
      \tcc*[l]{Generate the before-after mapping}
      \BFSD(\tcc*[h]{\textbf{MainVisitor}}){simplex $i \in \Closure(s)$}{
        \BFSU(\tcc*[h]{\textbf{InnerVisitor}}){simplex $j \in N \cap \Star(i)$}{
          \tcc*[l]{$i$ maps to $np$ so we need to connect $j$ to $np$ instead}
          Name newName = $np \cup j \setminus i$\;
          SimplexSet grab\;

          \BFSD(\tcc*[h]{\textbf{GrabVisitor}}){simplex $k \in N \cap \Closure(j)$}{
            \tcc*[l]{Grab dependent simplices which
                have not been grabbed yet. Grabbed simplices will map to simplex $newName$}
            N.\remove{k}\;
            grab.\insert{k}
          }
          \lIf(){$newName \notin simplexMap$}{
            simplexMap.\insert{pair(newName, grab)}
          }
          \lElse(){
            simplexMap[newName].\insert{grab}
          }
        }
      }
      \medskip
      \For(){$\{newName, grabbed\} \in simplexMap$}{
        \tcc*[l]{Run the user's callback to map simplex data.}
        DataType mappedData = (*clbk)(F, \name{j}, newName, grabbed)\;
        data.\insert{\{newName, mappedData\}}\tcc*[h]{Insert a pair containing new simplex name and data}
      }
      \tcc*[l]{Iterate over the complete neighborhood and remove simplices}
      \pR{F, doomed}\;
      \tcc*[l]{Iterate over data and append mapped simplices and data}
      \pI{F, data}

      \caption{Decimate a simplex by collapsing it into a vertex.}
      \label{alg:decimationalgo}
      \end{algorithm}

      \par This decimation algorithm is a generalization of an edge collapse operation to arbitrary dimensions. It is formally defined as follows:
      \begin{definition}
        Given simplicial complex $\mc{F}$, simplex to decimate $s \in \mc{F}$, vertex set $V$ of $\mc{F}$, and new vertex $p \notin V$, we define function,
        \begin{equation}
        \varphi(f) =
          \begin{cases}
            f             & \mbox{if } f \cap s = \varnothing\\
            p \cup (f \setminus s)  & \mbox{if } f \cap s \neq \varnothing
          \end{cases},
        \end{equation}
        where $f$ is any simplex $f \in \mc{F}$. We define the decimation of $\mc{F}$ by replacing $s$ with $p$ as $\varphi(\mc{F})$.
        \label{def:decimation}
      \end{definition}
      Note that decimation under this definition is not guaranteed to preserve the topology, as can by seen by decimating any edge or face of a tetrahedron.

      \par Decimation of a simplicial complex must result in a valid triangulation. Here we show that decimation by Definition~\ref{def:decimation} produces a valid abstract simplicial complex.
      \begin{theorem}
        $\varphi(\mc{F})$ is an abstract simplicial complex.
      \end{theorem}
      \begin{proof}
        Given simplices $y$ and $x$, let $y \in \varphi(\mc{F})$ and let $x \subset y$. We will show that $x \in \varphi(\mc{F})$.
        \smallskip
        There are two cases for $y \cap p$, as they can be disjoint or intersecting.

        Considering the disjoint case where $y \cap p = \varnothing$.
        This implies that $y \in \mc{F}$ and $y \cap s = \varnothing$.
        Since $\mc{F}$ is a simplicial complex, $y \in \mc{F}$ implies that $x \in \mc{F}$.
        Furthermore since $y \cap s = \varnothing$ and $x \subset y$, then $x \cap s = \varnothing$ implying that $\varphi(x) = x$ and thus $x \in \varphi(\mc{F})$.

        Alternately in the intersecting case where $y \cap p \neq \varnothing$, that is $y \cap p = p$.
        Since $x \subset y$ there are two sub-cases where $x$ either contains $p$ or does not.

        Supposing that $x \cap p = \varnothing$.
        Then $x \subset (f \setminus s)$ implying that $x \in \mc{F}$ and $x \cap s = \varnothing$.
        Therefore by the disjoint case, $\varphi(x) = x$ implies that $x \in \varphi(\mc{F})$.

        Supposing that $x \cap p = p$.
        We can rewrite any such simplex $x$ as $x = w \cup p$ where $w \cap p = \varnothing$.
        Furthermore we can write that $y = x \cup r$ where $x \cap r = \varnothing$ and $r\cap p = \varnothing$ and thus $y = w\cup p \cup r$.
        There exists some set $q$ such that $f = w \cup r \cup q$ and $q \subseteq s$ such that $\varphi(f) = p\cup ((w\cup r \cup q)\setminus s) = y$
        Since $f\in \mc{F}$ then $w\cup q \in \mc{F}$.
        Therefore $\varphi(w\cup q) = p\cup((w\cup q)\setminus s) = p \cup w = x$ and thus $x \in \varphi(\mc{F})$.

        For all cases and sub-cases we have shown that $x \in \varphi(\mc{F})$ therefore $\varphi(\mc{F})$ is an abstract simplicial complex.
      \end{proof}

    \begin{table}[ht!]
      \caption{Traversal order of the visitors for the decimation shown in Figure~\ref{fig:decexample}.}
       \begin{tabular}{@{}ccccc@{}}
                \toprule
                Order   & \verb|MainVisitor|& \verb|InnerVisitor|   & \verb|GrabVisitor|            & Maps to   \\ \midrule
                1       & \{3,4\}           & \{3,4\}               & \{3,4\}, \{3\}, \{4\}         & \{6\}     \\
                2       & \{3,4\}           & \{1,3,4\}             & \{1,3,4\}, \{1,3\}, \{1,4\}   & \{1,6\}   \\
                3       & \{3,4\}           & \{3,4,5\}             & \{3,4,5\}, \{3,5\}, \{4,5\}   & \{3,6\}   \\
                4       & \{3\}             & \{0,3\}               & \{0,3\}                       & \{0,6\}   \\
                5       & \{3\}             & \{0,1,3\}             & \{0,1,3\}                     & \{0,1,6\} \\
                6       & \{3\}             & \{0,3,5\}             & \{0,3,5\}                     & \{0,5,6\} \\
                7       & \{4\}             & \{2,4\}               & \{3,5\}                       & \{5,6\}   \\
                8       & \{4\}             & \{1,2,4\}             & \{1,2,4\}                     & \{1,2,6\} \\
                9       & \{4\}             & \{2,4,5\}             & \{2,4,5\}                     & \{2,5,6\} \\
            \bottomrule
            \end{tabular}
      \label{tab:decimateorder}
    \end{table}

    \par A pseudocode implementation for this decimation is provided in Algorithm~\ref{alg:decimationalgo}.
    Given some simplicial complex $\mc{F}$ and simplex $s \in \mc{F}$ to decimate, this algorithm works in four steps.
    First, we compute the complete neighborhood, $nbhd = \Star(\Closure(s))$, of $s$. Simplices not in the complete neighborhood will be invariant under $\varphi$ and are ignored.
    Next, we use a nested set of breadth first searches to walk over the complete neighborhood and compute $p \cup f \setminus s$ for each simplex in the neighborhood.
    The results are inserted into a \verb|SimplexMap| which maps $\varphi(f)$ to a \verb|SimplexSet| of all $f$ which map to $\varphi(f)$.
    Third, we iterate over the \verb|SimplexMap| and run the user defined callback on each mapping to generate a list of new simplices and associated mapped data stored in \verb|SimplexDataSet|.
    Finally, the algorithm removes all simplices in the complete neighborhood and inserts the new mapped simplices.

    \par An example application of this decimation operation is shown in Figure~\ref{fig:decimation}.
    In Figure~\ref{fig:decexample} we show the geometric realization of the complex before and after the decimation of simplex \{3,4\}.
    In Figure~\ref{fig:decdetail} we show the detailed Hasse diagrams for the constructed example. Note that there are two possible mapping situations.
    In one case, $f \in \Star(\Closure(s)) \cap \Closure(\Star(s))$, groups of simplices are merged.
    In the other case, simplices $f \in \Closure(\Star(s)) \setminus \Star(\Closure(s))$ only need to be reconnected to the new merged simplices.
    By carefully choosing the traversal order, some optimizations can be made.

    \par We apply the decimation on the constructed example shown in Figure~\ref{fig:decimation} and show the order of operations with respect to the current visited simplex for each visitor function in Table~\ref{tab:decimateorder}.
    Starting out, \verb|MainVisitor| and \verb|InnerVisitor| will visit \{3,4\}.
    At this point, \verb|GrabVisitor| will search BFS down from \{3,4\} to grab the set \{\{3,4\}, \{3\}, \{4\}\} and remove it from the neighborhood, eliminating some future calculations at \{3\} and \{4\}.
    All simplices in this set will map to new simplex \{6\} after decimation.
    Continuing upwards, \verb|InnerVisitor| will find simplex \{1,3,4\} and \verb|GrabVisitor| will grab set \{\{1,3,4\}, \{1,3\}, \{1,4\}\}.
    Again this set of simplices will map to common simplex \{1,6\} post decimation.
    A similar case occurs with simplex \{3,4,5\}. At this point, all simplices in $\Star(\Closure(s)) \cap \Closure(\Star(s))$ have been visited and removed from the neighborhood and \verb|MainVisitor| continues BFS down and finds \{3\} and calls BFS up (\verb|InnerVisitor|).
    Note that since simplex \{3\} has already been grabbed, \verb|InnerVisitor| will continue upwards and find \{0,3\}.
    Looking down there are no simplices which are faces of \{0,3\} in the neighborhood. So on and so forth.

    \par To reiterate, \verb|GrabVisitor| grabs the set of simplices which will be mapped to a common simplex.
    We show here that the order in which simplices are grabbed by Algorithm~\ref{alg:decimationalgo} will preserve that all simplices $f = w\cup q$ where $q\subseteq s$ will map to $\varphi(f) = w\cup p$.

    \par When visiting any simplex $f\supsetneq q$ where $q\subseteq s$ and $q$ corresponds to simplices visited by \verb|MainVisitor|.
    We can write $f$ as $f = w\cup q$ where the sets $w$ and $q$ are disjoint.
    Looking down from $f$ all simplices fall into two cases: $g = v\cup q$ where $\varnothing \subseteq v\subsetneq w$ or $h = w\cup t$ where $t\subsetneq q$.
    All simplices of form $g$, at worst case, will been grabbed while \verb|InnerVisitor| proceeded BFS up from $q$.
    Remaining simplices $h$ can be grouped with $f$ and correctly mapped to $w\cup p$.

    \par We note that in some non-manifold cases \verb|GrabVisitor| will not always grab set members in one visit.
    Supposing that we removed simplex \{1,3,4\} from the constructed example, in this case, \verb|InnerVisitor| cannot visit \{1,3,4\} and simplices \{1,3\} and \{1,4\} will not be grouped.
    Instead \{1,3\} and \{1,4\} will be found individually when \verb|MainVisitor| visits \{3\} then \{4\}.
    To catch this case and correctly map \{1,3\} and \{1,4\} to \{1,6\}, we use a \verb|SimplexMap| to aggregate all maps prior to proceeding.
    We note that in all cases starting with a valid simplicial complex, this implementation of the general collapse of simplex $s$ visits each member in $\Star(\Closure(s))$ and maps according to Def.~\ref{def:decimation} producing a valid simplicial complex.
    There is no guarantee that the result will have the same topological type as the pre-decimated mesh.
    The preservation of the topological type under decimation is often a desirable trait.
    We will show how to verify the Link Condition for when edge collapse of a 2-manifold will preserve the topological type in \S\ref{sec:preservingtopology}.

\section{Surface Mesh Application Example}\label{sec:applications}

\par \asc is a general simplicial complex data structure which is suitable for use in mesh manipulation and processing. For example, we can use \asc as the underlying representation for an orientable surface mesh. Using a predefined \verb|Vertex| class which is wrapped around a tensor library, and a class \verb|Orientable| which wraps an integer, we can easily create a surface mesh embedded in $\mathbb{R}^3$.
\begin{lstlisting}
using Vector = tensor<double,3,1>;
struct Vertex {
  Vector position;
  // ... other helpful vertex functions;
};
struct Orientable {
  int orientation;
};
struct complex_traits
{
    using KeyType = int;
    using NodeTypes = util::type_holder<void,Vertex,void,Orientable>;
    using EdgeTypes = util::type_holder<Orientable,Orientable,Orientable>;
};
using SurfaceMesh = simplicial_complex<complex_traits>;
\end{lstlisting}
In this case, \simplices{1} will store a \verb|Vertex| type while faces and all edges will store \verb|Orientable|. Using \verb|SurfaceMesh| we can easily create functions to load or write common mesh filetypes such as OFF as shown in the included library examples.

\begin{algorithm}[ht!]
\KwIn{Simplices $a$ and $b$ where $b = a\cup v$ and $v$ is a vertex.}
\KwOut{The orientation of edge $a\rightarrow b$}

\SetStartEndCondition{ (}{)}{)}
\AlgoDisplayBlockMarkers\AlgoDisplayGroupMarkers\SetAlgoBlockMarkers{ \{}{ \}\ }%
\SetAlgoNoEnd\SetAlgoNoLine\DontPrintSemicolon
\SetStartEndCondition{ (}{)}{)}
\SetKwFor{For}{for}{}{}
\SetKwIF{If}{ElseIf}{Else}{if}{}{elif}{else}{}

\medskip
int orient = 1\;
\For(\tcc*[h]{For each vertex $u$ in simplex $a$}){Vertex u : $a$}{
  \lIf(){v>a}{
    orient *= -1;
  }
}
\Return orient;

\caption{Define the orientation of a topological relation.}
\label{alg:orient}
\end{algorithm}

\par We can define a boundary morphism which applies on an ordered \simplex{k},
\begin{equation}
  \partial_i^k([a_0,\ldots, a_{k-1}]) = (-1)^{i}([a_0,\ldots, a_{k-1}]\setminus \{a_i\}),
\end{equation}
where $a_i < a_{i+1}$. Using Algorithm~\ref{alg:orient}, we can apply this morphism to assign a $\pm 1$ orientation to each topological relation in the complex. Subsequently, for orientable manifolds, we can compute orientations of faces $f_1$ and $f_2$ which share edge $e$ such that,
\begin{equation}
  \Orient(e_1)\cdot\Orient(f_1) + \Orient(e_2)\cdot\Orient(f_2) = 0,
\end{equation}
where $e_1$ and $e_2$ correspond to the edge up from $e$ to $f_1$ and $f_2$ respectively. Doing so, we create an oriented simplicial complex.

\par Supposing that we wish to compute the tangent of a vertex as defined by the weighted average tangent of incident faces. This is equivalent to computing the oriented wedge products of each incident face. This can be written generally as,
\begin{align}
  \Tangent(v) & = \frac{1}{N}\sum_{i=0}^{N} \Orient(f_i)\cdot(\partial_{j}(f_i)\largewedge \partial_{k}(f_i))\\
  & = \frac{1}{N}\sum_{i=0}^{N} \Orient(f_i)\cdot \frac{1}{2}(e_{i,j}\otimes e_{i,k} - e_{i,k}\otimes e_{i,j}),
\end{align}
where $N$ is the number of incident faces, $f_i$ is incident face $i$, $j$ and $k$ are indices of vertex members of $f_i$ not equal to $v$, and $e_{i,j} = \partial_{j}(f_i)$. This can be easily computed using a templated recursive function.
\begin{lstlisting}
// Terminal case
auto getTangentH(const SurfaceMesh& mesh,
    const Vector& origin,
    SurfaceMesh::SimplexID<SurfaceMesh::topLevel> curr){
    return (*curr).orientation;
}

template <std::size_t level, std::size_t dimension>
auto getTangentH(const SurfaceMesh& mesh,
    const Vector& origin,
    SurfaceMesh::SimplexID<level> curr){
    tensor<double, 3, SurfaceMesh::topLevel - level> rval;
    auto cover = mesh.get_cover(curr); // Lookup coboundary relations
    for(auto v : cover){
        auto edge = *mesh.get_edge_up(curr, v); // Get the edge object
        const auto& vtx = (*mesh.get_simplex_up({v})).position; // Vertex v
        auto next = mesh.get_simplex_up(curr,v); // Simplex curr union v
        rval += edge.orientation * (v-origin) * getTangentH(mesh, origin, next);
    }
    return rval/cover.size();
}
\end{lstlisting}
This demonstrates the ease using the \asc library as an underlying simplicial complex representation. Using the provided API, it is easy to traverse the complex to perform any computations.

\subsection{Preservation of Topology Type of Surface Mesh Under Edge Decimation by Contraction}\label{sec:preservingtopology}
  \par Supposing we wish to decimate a surface mesh by edge contraction under Def.~\ref{def:decimation} without changing the topology of the complex.
  This can be verified by checking the Link Condition, defined with proof from Edelsbrunner\cite{Edelsbrunner2001}, stating,
  \begin{lemma}
    Let $F$ be a triangulation of a 2-manifold. The contraction of $ab\in F$ preserves the topological type if and only if $\Link(a)\cap\Link(b) = \Link(ab)$.
    \label{lemma:linkcondition}
  \end{lemma}

  \par Revisiting the example from Fig.~\ref{fig:decexample}, we can construct the topology and check the Link Condition using operations supported by the \asc library.
\begin{lstlisting}
// Construct the topology
SurfaceMesh mesh;
mesh.insert({0,1,3});
mesh.insert({0,3,5});
mesh.insert({1,3,4});
mesh.insert({3,4,5});
mesh.insert({1,2,4});
mesh.insert({2,4,5});

// Get simplices and edge
SimplexSet<SurfaceMesh> A,B,AB, AcapB;
auto ab = mesh.get_simplex_up({3,4});
auto a = mesh.get_simplex_up({3});
auto b = mesh.get_simplex_up({4});

// Compute the links of each
getLink(mesh, a, A);
getLink(mesh, b, B);
getLink(mesh, ab, AB);

// Link(a): Simplices {1}, {4}, {5}, {1,4}, {0,1}, {0,5}, {4,5}
std::cout << A << std::endl;
// Link(b): Simplices {5}, {2}, {3}, {1}, {3,5}, {1,3}, {1,2}, {2,5}
std::cout << B << std::endl;
// Link(AB): Simplices {1}, {5}
std::cout << AB << std::endl;

set_intersection(A, B, AcapB);
// Link(a) \cap Link(b): Simplices {1}, {5}
std::cout << AcapB << std::endl;

// Check the Link Condition
std::cout << (AB == AcapB) << std::endl; // Evaluates to true
\end{lstlisting}

  The \verb|getLink()| function utilizes a series of visitor breadth first traversals down then up the complex, collecting the set of simplex into a \verb|SimplexSet| object.
  Then using set operations provided by \verb|SimplexSet| the set difference and equality comparison are performed.
  This example highlights the simplicity, clarity, and transparency of using the \asc library.

\section{Conclusions}

  \par \asc provides a general simplicial complex data structure which allows the storage of user defined types at each simplex level.
  The library comes with a full-featured API providing common simplicial complex operations, as well as support for complex traversals using a visitor.
  We also provide a metadata aware decimation algorithm which allows users to collapse simplices of any dimension while preserving data according to a user defined mapping function.
  Our implementation of \asc using a strongly-typed language is only possible due to recent innovations in language tools.
  The \asc API abstracts away most of the complicated templating, allowing it to be both modern and easy to use.
  We anticipate that \asc will not only be of use for the ET community but microscopy and modeling as a whole along with other fields like applied mathematics and CAD.

  \par One limitation is the ease of extending to other languages.
  The generality of \asc is reliant upon the C++ compiler.
  Sacrificing this, specific realizations of \asc can be wrapped using tools like SWIG for use in other languages.
  Another limitation of \asc is the memory efficiency.
  The current Hasse diagram based implementation was selected for the sake of transparency, ease of traversal, and manipulation.
  Optimizations to the memory efficiency of \asc can be made by employing more compact representations.
  The variadic template approach we use to attach user data to simplices is compatible with data structures which explicitly represent all simplices but only a subset of topological relations.
  This includes data structures such as SimplexTree\cite{Boissonnat2012}, IS\cite{DeFloriani2010b}, and SIG\cite{DeFloriani2004} among others.
  Other compressed data structures which skip levels of low importance using implicit nodes are not compatible with the current \asc implementation.
  Skipped levels would need to be implemented as exceptions to the combinatorial variadic rules.
  Similarly, although \asc in its current form is restricted to the representation of simplicial complexes, the combinatorial strategy can be easily adapted to support other regular polytopes by changing the boundary relation storage rules.
  In the future we hope to incorporate parallelism into the \asc library.
  A copy of \asc along with online documentation can found on GitHub \url{https://github.com/ctlee/casc}.

\begin{printonly}
  \appendix
  \section{Supplementary Materials}
  Please see the supplementary materials in the online version.
\end{printonly}

\begin{acks}
The authors would like to thank the anonymous referees for
their valuable comments and helpful suggestions.
This work is supported by the \grantsponsor{100000057}{National Institutes of Health, NIGMS}{http://dx.doi.org/10.13039/100000057} under grant numbers \grantnum[]{100000057}{P41-GM103426} and \grantnum[]{100000057}{RO1-GM31749}.
CTL also acknowledges support from the NIH Molecular Biophysics Training Grant \grantnum[]{100000057}{T32-GM008326}.
MJH was supported in part by the \grantsponsor{NSFDMS}{National Science Foundation, Division of Mathematical Sciences}{http://dx.doi.org/10.13039/100000121} under awards \grantnum[]{NSFDMS}{DMS-CM1620366} and \grantnum[]{NSFDMS}{DMS-FRG1262982}.
\end{acks}


\begin{thebibliography}{24}


\ifx \showCODEN    \undefined \def \showCODEN     #1{\unskip}     \fi
\ifx \showDOI      \undefined \def \showDOI       #1{#1}\fi
\ifx \showISBNx    \undefined \def \showISBNx     #1{\unskip}     \fi
\ifx \showISBNxiii \undefined \def \showISBNxiii  #1{\unskip}     \fi
\ifx \showISSN     \undefined \def \showISSN      #1{\unskip}     \fi
\ifx \showLCCN     \undefined \def \showLCCN      #1{\unskip}     \fi
\ifx \shownote     \undefined \def \shownote      #1{#1}          \fi
\ifx \showarticletitle \undefined \def \showarticletitle #1{#1}   \fi
\ifx \showURL      \undefined \def \showURL       {\relax}        \fi
\providecommand\bibfield[2]{#2}
\providecommand\bibinfo[2]{#2}
\providecommand\natexlab[1]{#1}
\providecommand\showeprint[2][]{arXiv:#2}

\bibitem[\protect\citeauthoryear{??}{CGA}{[n. d.]}]%
        {CGAL}
 \bibinfo{year}{[n. d.]}\natexlab{}.
\newblock \bibinfo{title}{{CGAL, Computational Geometry Algorithms Library}}.
\newblock
\newblock
\urldef\tempurl%
\url{http://www.cgal.org}
\showURL{%
\tempurl}


\bibitem[\protect\citeauthoryear{Babu{\v{s}}ka and Aziz}{Babu{\v{s}}ka and
  Aziz}{1976}]%
        {Babuska1976}
\bibfield{author}{\bibinfo{person}{I. Babu{\v{s}}ka} {and}
  \bibinfo{person}{A.~K. Aziz}.} \bibinfo{year}{1976}\natexlab{}.
\newblock \showarticletitle{{On the Angle Condition in the Finite Element
  Method}}.
\newblock \bibinfo{journal}{\emph{SIAM J. Numer. Anal.}} \bibinfo{volume}{13},
  \bibinfo{number}{2} (\bibinfo{date}{apr} \bibinfo{year}{1976}),
  \bibinfo{pages}{214--226}.
\newblock
\showISSN{0036-1429}
\urldef\tempurl%
\url{https://doi.org/10.1137/0713021}
\showDOI{\tempurl}


\bibitem[\protect\citeauthoryear{Bank, Sherman, and Weiser}{Bank
  et~al\mbox{.}}{1983}]%
        {Bank1983}
\bibfield{author}{\bibinfo{person}{Randolph~E. Bank},
  \bibinfo{person}{Andrew~H. Sherman}, {and} \bibinfo{person}{Alan Weiser}.}
  \bibinfo{year}{1983}\natexlab{}.
\newblock \showarticletitle{{Refinement Algorithms and Data Structures for
  Regular Local Mesh Refinement}}.
\newblock In \bibinfo{booktitle}{\emph{Sci. Comput. Appl. Methematics Comput.
  to Phys. Sci.}}, \bibfield{editor}{\bibinfo{person}{R.~S. Stepleman}} (Ed.).
  \bibinfo{publisher}{North-Holland}, \bibinfo{pages}{3--17}.
\newblock


\bibitem[\protect\citeauthoryear{Bank and Xu}{Bank and Xu}{1996}]%
        {Bank1996}
\bibfield{author}{\bibinfo{person}{Randolph~E. Bank} {and}
  \bibinfo{person}{Jinchao Xu}.} \bibinfo{year}{1996}\natexlab{}.
\newblock \showarticletitle{{An Algorithm for Coarsening Unstructured Meshes}}.
\newblock \bibinfo{journal}{\emph{Numer. Math.}} \bibinfo{volume}{73},
  \bibinfo{number}{1} (\bibinfo{date}{mar} \bibinfo{year}{1996}),
  \bibinfo{pages}{1--36}.
\newblock
\showISSN{0029-599X}
\urldef\tempurl%
\url{https://doi.org/10.1007/s002110050181}
\showDOI{\tempurl}


\bibitem[\protect\citeauthoryear{Boissonnat, Devillers, and Hornus}{Boissonnat
  et~al\mbox{.}}{2009}]%
        {boissonnat:inria-00412437}
\bibfield{author}{\bibinfo{person}{Jean-Daniel Boissonnat},
  \bibinfo{person}{Olivier Devillers}, {and} \bibinfo{person}{Samuel Hornus}.}
  \bibinfo{year}{2009}\natexlab{}.
\newblock \showarticletitle{{Incremental Construction of the Delaunay Graph in
  Medium Dimension}}. In \bibinfo{booktitle}{\emph{Annu. Symp. Comput. Geom.}}
  \bibinfo{address}{Aarhus, Denmark}, \bibinfo{pages}{208--216}.
\newblock
\urldef\tempurl%
\url{https://hal.inria.fr/inria-00412437}
\showURL{%
\tempurl}


\bibitem[\protect\citeauthoryear{Boissonnat and Maria}{Boissonnat and
  Maria}{2014}]%
        {Boissonnat2012}
\bibfield{author}{\bibinfo{person}{Jean-Daniel Boissonnat} {and}
  \bibinfo{person}{Cl{\'{e}}ment Maria}.} \bibinfo{year}{2014}\natexlab{}.
\newblock \showarticletitle{{The Simplex Tree: An Efficient Data Structure for
  General Simplicial Complexes}}.
\newblock \bibinfo{journal}{\emph{Algorithmica}} \bibinfo{volume}{70},
  \bibinfo{number}{3} (\bibinfo{date}{nov} \bibinfo{year}{2014}),
  \bibinfo{pages}{406--427}.
\newblock
\showISSN{0178-4617}
\urldef\tempurl%
\url{https://doi.org/10.1007/s00453-014-9887-3}
\showDOI{\tempurl}


\bibitem[\protect\citeauthoryear{Canino, {De Floriani}, and Weiss}{Canino
  et~al\mbox{.}}{2011}]%
        {Canino2011}
\bibfield{author}{\bibinfo{person}{David Canino}, \bibinfo{person}{Leila {De
  Floriani}}, {and} \bibinfo{person}{Kenneth Weiss}.}
  \bibinfo{year}{2011}\natexlab{}.
\newblock \showarticletitle{{IA*: An Adjacency-Based Representation for
  Non-Manifold Simplicial Shapes in Arbitrary Dimensions}}.
\newblock \bibinfo{journal}{\emph{Comput. Graph.}} \bibinfo{volume}{35},
  \bibinfo{number}{3} (\bibinfo{date}{jun} \bibinfo{year}{2011}),
  \bibinfo{pages}{747--753}.
\newblock
\showISSN{00978493}
\urldef\tempurl%
\url{https://doi.org/10.1016/j.cag.2011.03.009}
\showDOI{\tempurl}


\bibitem[\protect\citeauthoryear{{De Floriani}, Greenfieldboyce, and Hui}{{De
  Floriani} et~al\mbox{.}}{2004}]%
        {DeFloriani2004}
\bibfield{author}{\bibinfo{person}{Leila {De Floriani}}, \bibinfo{person}{David
  Greenfieldboyce}, {and} \bibinfo{person}{Annie Hui}.}
  \bibinfo{year}{2004}\natexlab{}.
\newblock \showarticletitle{{A Data Structure for Non-Manifold Simplicial D
  -Complexes}}. In \bibinfo{booktitle}{\emph{Proc. 2004 Eurographics/ACM
  SIGGRAPH Symp. Geom. Process. - SGP '04}}. \bibinfo{publisher}{ACM Press},
  \bibinfo{address}{New York, New York, USA}, \bibinfo{pages}{83}.
\newblock
\showISBNx{3905673134}
\showISSN{17278384}
\urldef\tempurl%
\url{https://doi.org/10.1145/1057432.1057444}
\showDOI{\tempurl}


\bibitem[\protect\citeauthoryear{{De Floriani} and Hui}{{De Floriani} and
  Hui}{2005}]%
        {DeFloriani2005}
\bibfield{author}{\bibinfo{person}{Leila {De Floriani}} {and}
  \bibinfo{person}{Annie Hui}.} \bibinfo{year}{2005}\natexlab{}.
\newblock \showarticletitle{{Data Structures for Simplicial Complexes: An
  Analysis and a Comparison}}. In \bibinfo{booktitle}{\emph{Proc. Third
  Eurographics Symp. Geom. Process.}} \bibinfo{publisher}{Eurographics
  Association}, \bibinfo{address}{Vienna}, \bibinfo{pages}{119}.
\newblock
\urldef\tempurl%
\url{http://dl.acm.org/citation.cfm?id=1281920.1281940}
\showURL{%
\tempurl}


\bibitem[\protect\citeauthoryear{{De Floriani}, Hui, Panozzo, and Canino}{{De
  Floriani} et~al\mbox{.}}{2010}]%
        {DeFloriani2010b}
\bibfield{author}{\bibinfo{person}{Leila {De Floriani}}, \bibinfo{person}{Annie
  Hui}, \bibinfo{person}{Daniele Panozzo}, {and} \bibinfo{person}{David
  Canino}.} \bibinfo{year}{2010}\natexlab{}.
\newblock \showarticletitle{{A Dimension-Independent Data Structure for
  Simplicial Complexes}}.
\newblock In \bibinfo{booktitle}{\emph{Proc. 19th Int. Meshing Roundtable}}.
  \bibinfo{publisher}{Springer Berlin Heidelberg}, \bibinfo{address}{Berlin,
  Heidelberg}, \bibinfo{pages}{403--420}.
\newblock
\urldef\tempurl%
\url{https://doi.org/10.1007/978-3-642-15414-0_24}
\showDOI{\tempurl}


\bibitem[\protect\citeauthoryear{Dyedov, Ray, Einstein, Jiao, and
  Tautges}{Dyedov et~al\mbox{.}}{2015}]%
        {Dyedov2015}
\bibfield{author}{\bibinfo{person}{Vladimir Dyedov}, \bibinfo{person}{Navamita
  Ray}, \bibinfo{person}{Daniel Einstein}, \bibinfo{person}{Xiangmin Jiao},
  {and} \bibinfo{person}{Timothy~J. Tautges}.} \bibinfo{year}{2015}\natexlab{}.
\newblock \showarticletitle{{AHF: Array-Based Half-Facet Data Structure for
  Mixed-Dimensional and Non-Manifold Meshes}}.
\newblock \bibinfo{journal}{\emph{Eng. Comput.}} \bibinfo{volume}{31},
  \bibinfo{number}{3} (\bibinfo{date}{jul} \bibinfo{year}{2015}),
  \bibinfo{pages}{389--404}.
\newblock
\showISSN{0177-0667}
\urldef\tempurl%
\url{https://doi.org/10.1007/s00366-014-0378-6}
\showDOI{\tempurl}


\bibitem[\protect\citeauthoryear{Edelsbrunner}{Edelsbrunner}{2001}]%
        {Edelsbrunner2001}
\bibfield{author}{\bibinfo{person}{Herbert Edelsbrunner}.}
  \bibinfo{year}{2001}\natexlab{}.
\newblock \bibinfo{booktitle}{\emph{{Geometry and Topology for Mesh
  Generation}}}.
\newblock \bibinfo{publisher}{Cambridge University Press},
  \bibinfo{address}{Cambridge}.
\newblock
\showISBNx{9780511530067}
\urldef\tempurl%
\url{https://doi.org/10.1017/CBO9780511530067}
\showDOI{\tempurl}


\bibitem[\protect\citeauthoryear{Gao, Yu, and Holst}{Gao et~al\mbox{.}}{2013}]%
        {Gao2013}
\bibfield{author}{\bibinfo{person}{Zhanheng Gao}, \bibinfo{person}{Zeyun Yu},
  {and} \bibinfo{person}{Michael Holst}.} \bibinfo{year}{2013}\natexlab{}.
\newblock \showarticletitle{{Feature-Preserving Surface Mesh Smoothing Via
  Suboptimal Delaunay Triangulation}}.
\newblock \bibinfo{journal}{\emph{Graph. Models}} \bibinfo{volume}{75},
  \bibinfo{number}{1} (\bibinfo{date}{jan} \bibinfo{year}{2013}),
  \bibinfo{pages}{23--38}.
\newblock
\showISSN{15240703}
\urldef\tempurl%
\url{https://doi.org/10.1016/j.gmod.2012.10.007}
\showDOI{\tempurl}


\bibitem[\protect\citeauthoryear{Holst, Sarbach, Tiglio, and Vallisneri}{Holst
  et~al\mbox{.}}{2016}]%
        {HSTV16a}
\bibfield{author}{\bibinfo{person}{M. Holst}, \bibinfo{person}{O. Sarbach},
  \bibinfo{person}{M. Tiglio}, {and} \bibinfo{person}{M. Vallisneri}.}
  \bibinfo{year}{2016}\natexlab{}.
\newblock \showarticletitle{The Emergence of Gravitational Wave Science: 100
  Years of Development of Mathematical Theory, Detectors, Numerical Algorithms,
  and Data Analysis Tools}.
\newblock \bibinfo{journal}{\emph{Bull. Amer. Math. Soc.}}
  \bibinfo{volume}{53} (\bibinfo{year}{2016}), \bibinfo{pages}{513--554}.
\newblock
\urldef\tempurl%
\url{http://dx.doi.org/10.1090/bull/1544}
\showURL{%
\tempurl}


\bibitem[\protect\citeauthoryear{Holst and Tiee}{Holst and Tiee}{2018}]%
        {HoTi14a}
\bibfield{author}{\bibinfo{person}{M. Holst} {and} \bibinfo{person}{C. Tiee}.}
  \bibinfo{year}{2018}\natexlab{}.
\newblock \showarticletitle{Finite Element Exterior Calculus for Parabolic
  Evolution Problems on {Riemannian} Hypersurfaces}.
\newblock \bibinfo{journal}{\emph{Journal of Computational Mathematics}}
  \bibinfo{volume}{36}, \bibinfo{number}{6} (\bibinfo{year}{2018}),
  \bibinfo{pages}{792--832}.
\newblock


\bibitem[\protect\citeauthoryear{Liu and Joe}{Liu and Joe}{1994}]%
        {Liu1994}
\bibfield{author}{\bibinfo{person}{A. Liu} {and} \bibinfo{person}{B. Joe}.}
  \bibinfo{year}{1994}\natexlab{}.
\newblock \showarticletitle{{Relationship Between Tetrahedron Shape Measures}}.
\newblock \bibinfo{journal}{\emph{BIT}} \bibinfo{volume}{34},
  \bibinfo{number}{2} (\bibinfo{date}{jun} \bibinfo{year}{1994}),
  \bibinfo{pages}{268--287}.
\newblock
\showISSN{0006-3835}
\urldef\tempurl%
\url{https://doi.org/10.1007/BF01955874}
\showDOI{\tempurl}


\bibitem[\protect\citeauthoryear{Ray, Grindeanu, Zhao, Mahadevan, and Jiao}{Ray
  et~al\mbox{.}}{2015}]%
        {Ray2015}
\bibfield{author}{\bibinfo{person}{Navamita Ray}, \bibinfo{person}{Iulian
  Grindeanu}, \bibinfo{person}{Xinglin Zhao}, \bibinfo{person}{Vijay
  Mahadevan}, {and} \bibinfo{person}{Xiangmin Jiao}.}
  \bibinfo{year}{2015}\natexlab{}.
\newblock \showarticletitle{{Array-Based Hierarchical Mesh Generation in
  Parallel}}.
\newblock \bibinfo{journal}{\emph{Procedia Eng.}}  \bibinfo{volume}{124}
  (\bibinfo{year}{2015}), \bibinfo{pages}{291--303}.
\newblock
\showISSN{18777058}
\urldef\tempurl%
\url{https://doi.org/10.1016/j.proeng.2015.10.140}
\showDOI{\tempurl}


\bibitem[\protect\citeauthoryear{Regge}{Regge}{1961}]%
        {Regge1961}
\bibfield{author}{\bibinfo{person}{T. Regge}.} \bibinfo{year}{1961}\natexlab{}.
\newblock \showarticletitle{{General Relativity Without Coordinates}}.
\newblock \bibinfo{journal}{\emph{Nuovo Cim.}} \bibinfo{volume}{19},
  \bibinfo{number}{3} (\bibinfo{date}{feb} \bibinfo{year}{1961}),
  \bibinfo{pages}{558--571}.
\newblock
\showISSN{0029-6341}
\urldef\tempurl%
\url{https://doi.org/10.1007/BF02733251}
\showDOI{\tempurl}


\bibitem[\protect\citeauthoryear{Roberts}{Roberts}{2014}]%
        {Roberts2014}
\bibfield{author}{\bibinfo{person}{Elijah Roberts}.}
  \bibinfo{year}{2014}\natexlab{}.
\newblock \showarticletitle{{Cellular and Molecular Structure As a Unifying
  Framework for Whole-Cell Modeling}}.
\newblock \bibinfo{journal}{\emph{Curr. Opin. Struct. Biol.}}
  \bibinfo{volume}{25} (\bibinfo{date}{apr} \bibinfo{year}{2014}),
  \bibinfo{pages}{86--91}.
\newblock
\showISSN{0959440X}
\urldef\tempurl%
\url{https://doi.org/10.1016/j.sbi.2014.01.005}
\showDOI{\tempurl}


\bibitem[\protect\citeauthoryear{Si}{Si}{2015}]%
        {Si2015a}
\bibfield{author}{\bibinfo{person}{Hang Si}.} \bibinfo{year}{2015}\natexlab{}.
\newblock \showarticletitle{{TetGen, a Delaunay-Based Quality Tetrahedral Mesh
  Generator}}.
\newblock \bibinfo{journal}{\emph{ACM Trans. Math. Softw.}}
  \bibinfo{volume}{41}, \bibinfo{number}{2} (\bibinfo{date}{feb}
  \bibinfo{year}{2015}), \bibinfo{pages}{1--36}.
\newblock
\showISSN{00983500}
\urldef\tempurl%
\url{https://doi.org/10.1145/2629697}
\showDOI{\tempurl}


\bibitem[\protect\citeauthoryear{Vanden-Eijnden and Venturoli}{Vanden-Eijnden
  and Venturoli}{2009}]%
        {Vanden-Eijnden2009d}
\bibfield{author}{\bibinfo{person}{Eric Vanden-Eijnden} {and}
  \bibinfo{person}{Maddalena Venturoli}.} \bibinfo{year}{2009}\natexlab{}.
\newblock \showarticletitle{{Markovian Milestoning with Voronoi
  Tessellations}}.
\newblock \bibinfo{journal}{\emph{J. Chem. Phys.}} \bibinfo{volume}{130},
  \bibinfo{number}{19} (\bibinfo{date}{may} \bibinfo{year}{2009}),
  \bibinfo{pages}{194101}.
\newblock
\showISSN{0021-9606}
\urldef\tempurl%
\url{https://doi.org/10.1063/1.3129843}
\showDOI{\tempurl}


\bibitem[\protect\citeauthoryear{Yu, Holst, and {Andrew McCammon}}{Yu
  et~al\mbox{.}}{2008a}]%
        {Yu2008}
\bibfield{author}{\bibinfo{person}{Zeyun Yu}, \bibinfo{person}{Michael~J.
  Holst}, {and} \bibinfo{person}{J. {Andrew McCammon}}.}
  \bibinfo{year}{2008}\natexlab{a}.
\newblock \showarticletitle{{High-Fidelity Geometric Modeling for Biomedical
  Applications}}.
\newblock \bibinfo{journal}{\emph{Finite Elem. Anal. Des.}}
  \bibinfo{volume}{44}, \bibinfo{number}{11} (\bibinfo{date}{jul}
  \bibinfo{year}{2008}), \bibinfo{pages}{715--723}.
\newblock
\showISSN{0168874X}
\urldef\tempurl%
\url{https://doi.org/10.1016/j.finel.2008.03.004}
\showDOI{\tempurl}


\bibitem[\protect\citeauthoryear{Yu, Holst, Hayashi, Bajaj, Ellisman, McCammon,
  and Hoshijima}{Yu et~al\mbox{.}}{2008b}]%
        {Yu2008a}
\bibfield{author}{\bibinfo{person}{Zeyun Yu}, \bibinfo{person}{Michael~J
  Holst}, \bibinfo{person}{Takeharu Hayashi}, \bibinfo{person}{Chandrajit~L
  Bajaj}, \bibinfo{person}{Mark~H Ellisman}, \bibinfo{person}{J~Andrew
  McCammon}, {and} \bibinfo{person}{Masahiko Hoshijima}.}
  \bibinfo{year}{2008}\natexlab{b}.
\newblock \showarticletitle{{Three-Dimensional Geometric Modeling of
  Membrane-Bound Organelles in Ventricular Myocytes: Bridging the Gap Between
  Microscopic Imaging and Mathematical Simulation}}.
\newblock \bibinfo{journal}{\emph{J. Struct. Biol.}} \bibinfo{volume}{164},
  \bibinfo{number}{3} (\bibinfo{date}{dec} \bibinfo{year}{2008}),
  \bibinfo{pages}{304--313}.
\newblock
\showISSN{10478477}
\urldef\tempurl%
\url{https://doi.org/10.1016/j.jsb.2008.09.004}
\showDOI{\tempurl}


\bibitem[\protect\citeauthoryear{Zhang}{Zhang}{2016}]%
        {Zhang2016}
\bibfield{author}{\bibinfo{person}{Yongjie Zhang}.}
  \bibinfo{year}{2016}\natexlab{}.
\newblock \bibinfo{booktitle}{\emph{{Geometric Modeling and Mesh Generation
  from Scanned Images}}}.
\newblock \bibinfo{publisher}{Chapman and Hall/CRC}.
\newblock
\showISBNx{978-1-4822-2776-5}
\urldef\tempurl%
\url{https://doi.org/10.1201/b19466}
\showDOI{\tempurl}


\end{thebibliography}

\begin{screenonly}
\appendix
\section{Supplementary Materials}
  \subsection{Derivation of Worse Case Insertion Performance}
\par The number of combinations at each level $l$ for a $k$-simplex is given by the binomial theorem, $\binom{k}{l}$.
In the worst case, where the simplicial complex is empty, the total number of nodes created for inserting a $k$-simplex is $\sum_{l=0}^{k}\binom{k}{l} = 2^k$.
The number of edges to represent all topological relations is then given by $\sum_{l=0}^{k}l\binom{k}{l} = k\cdot 2^{k-1}$.

\subsection{Templated Insertion Algorithm}

\par Pseudocode for the algorithms presented in this manuscript have been vastly simplified in order to facilitate understanding.
For example Algorithm~\ref{alg:insertion}, while the non-templated version is appears straightforward, it is impossible to be implemented in C++ directly, due to several typing related issues.
First, the function prototype for \verb|insert()| requires the \verb|rootSimplex| as the second argument.
Simplices at different levels have different types and \verb|insert()| must be overloaded.
Similarly the variable \verb|newSimplex| and function \verb|createSimplex()| must know the type of simplex which will be created at compile time.

\par The actual implementation uses variadic templates to resolve the typing issues.
As an example, templated pseudocode for simplex insertion (Algorithm~\ref{alg:insertion}) is shown in Algorithm~\ref{alg:templatedinsertion}.
Not only does the templated code automatically build the correct overloaded functions, but it provides many optimizations.

\par The step-by-step insertion of tetrahedron \{1,2,3,4\} is shown in Figure~\ref{fig:insertorder}.
Numbered red lines correspond to \verb|newNode| and \verb|root| in function \verb|insertNode()|.
Skinny black lines are the topological relations inserted by \verb|backfill()|.

\newpage

\begin{algorithm}[ht!]
\KwIn{$keys$[$n$]: Indices of $n$ simplices to describe new simplex $s$,\\
$\mc{F}$: simplicial complex}
\KwOut{The new simplex $s$}

\SetStartEndCondition{ (}{)}{)}\SetAlgoBlockMarkers{}{\}}
\SetKwFor{For}{for}{\{}{}
\SetKwIF{If}{ElseIf}{Else}{if}{\{}{elif}{else\{}{}
\SetKw{KwInFor}{in}
\AlgoDisplayBlockMarkers\SetAlgoNoLine\DontPrintSemicolon

\makeatletter
\newcommand{\TSetKwFunction}[2]{%
  \expandafter\gdef\csname @#1\endcsname##1##2{
    \FuncSty{#2}\FuncSty{<}\FuncArgSty{##1}\FuncSty{>}\FuncSty{(}\FuncArgSty{##2}\FuncSty{)}
  }%
  \expandafter\gdef\csname#1\endcsname{
    \@ifnextchar\bgroup{\csname @#1\endcsname}{\FuncSty{#2}\xspace}
  }%
}%
\makeatother

\SetKwProg{Fn}{Function }{\{}{}
\TSetKwFunction{iFull}{setupForLoop}
\TSetKwFunction{iFor}{forLoop}
\TSetKwFunction{iRaw}{insertNode}
\TSetKwFunction{backfill}{backfill}
\TSetKwFunction{insert}{insert}
\TSetKwFunction{create}{createSimplex}
\medskip
\tcc*[l]{User function to insert simplex \{keys\}}
\Fn(){\insert{n}{keys[n]}}{
  \Return \iFull{0, n}{root, keys} \tcc*[h]{'$root$' is the root node}
}
\medskip
\tcp{The following are private library functions...}

\tcc{\textbf{Array slice operation.} Algorithm~\ref{alg:insertion}: keys[0:i]}
\Fn(\tcc*[h]{General template}){\iFull{level, n}{root, keys}}{
  \Return \iFor{level, n, n}{root, keys} \tcc*[h]{Setup the recursive for loop}\;
}
\Fn(\tcc*[h]{Terminal condition $n=0$}){\iFull{level, 0}{root, keys}}{
  \Return root\;
}
\medskip
\tcc{\textbf{Templated for loop.} Algorithm~\ref{alg:insertion}: for (i = 0; i < n; i++)}
\Fn(\tcc*[h]{For loop for antistep}){\iFor{level, antistep, n}{root, keys}}{
  \tcc*[l]{$n-antistep$ defines next key to add to $root$}
  \iRaw{level, n-antistep}{root, keys}\;
  \Return \iFor{level, antistep-1, n}{root, keys}\;
}
\Fn(\tcc*[h]{Stop when $antistep$ = 1}){\iFor{level, 1, n}{root, keys}}{
  \Return \iRaw{level, n-1}{root, keys}\;
}
\medskip
\Fn(\tcc*[h]{Insert a new node}){\iRaw{level, n}{root, keys}}{
  v = keys[n]\;
  \eIf(){$root\cup v \in \mc{F}$}{
    newNode= root.up[v]\;
  }(\tcc*[h]{Add simplex $root\cup v$}){
    newNode = \create{n}{} \tcc*[h]{Create a new node, \simplex{n} $newNode$}\;
    newNode.down[v] = root \tcc*[h]{Connect boundary relation}\;
    root.up[v] = newNode \tcc*[h]{Connect coboundary relation}\;

    \backfill(root, newNode, v)\tcc*[h]{Backfill other topological relations}\;
  }
  \tcc{\textbf{Recurse to insert any cofaces of newNode.} Algorithm~\ref{alg:insertion}: insert(keys[0:i], newSimplex)}
  \Return \iFull{level+1, n}{newNode, keys}
}
\medskip
\Fn(\tcc*[h]{Backfilling pointers to other parents}){\backfill{level}{root, newNode, value}}{
  \For{currentNode \KwInFor root.down}{
    childNode = currentNode.up[value] \tcc*[h]{Get simplex $currentNode \cup value$}\;

    newNode.down[value] = child \tcc*[h]{Connect boundary relation}\;
    child.up[value] = newNode \tcc*[h]{Connect coboundary relation}\;
  }
}
\caption{Templated pseudocode implementation of Algorithm~\ref{alg:insertion}.}
\label{alg:templatedinsertion}
\end{algorithm}

\newpage
\begin{figure}[H]
  \centering
  \includegraphics[width=\textwidth]{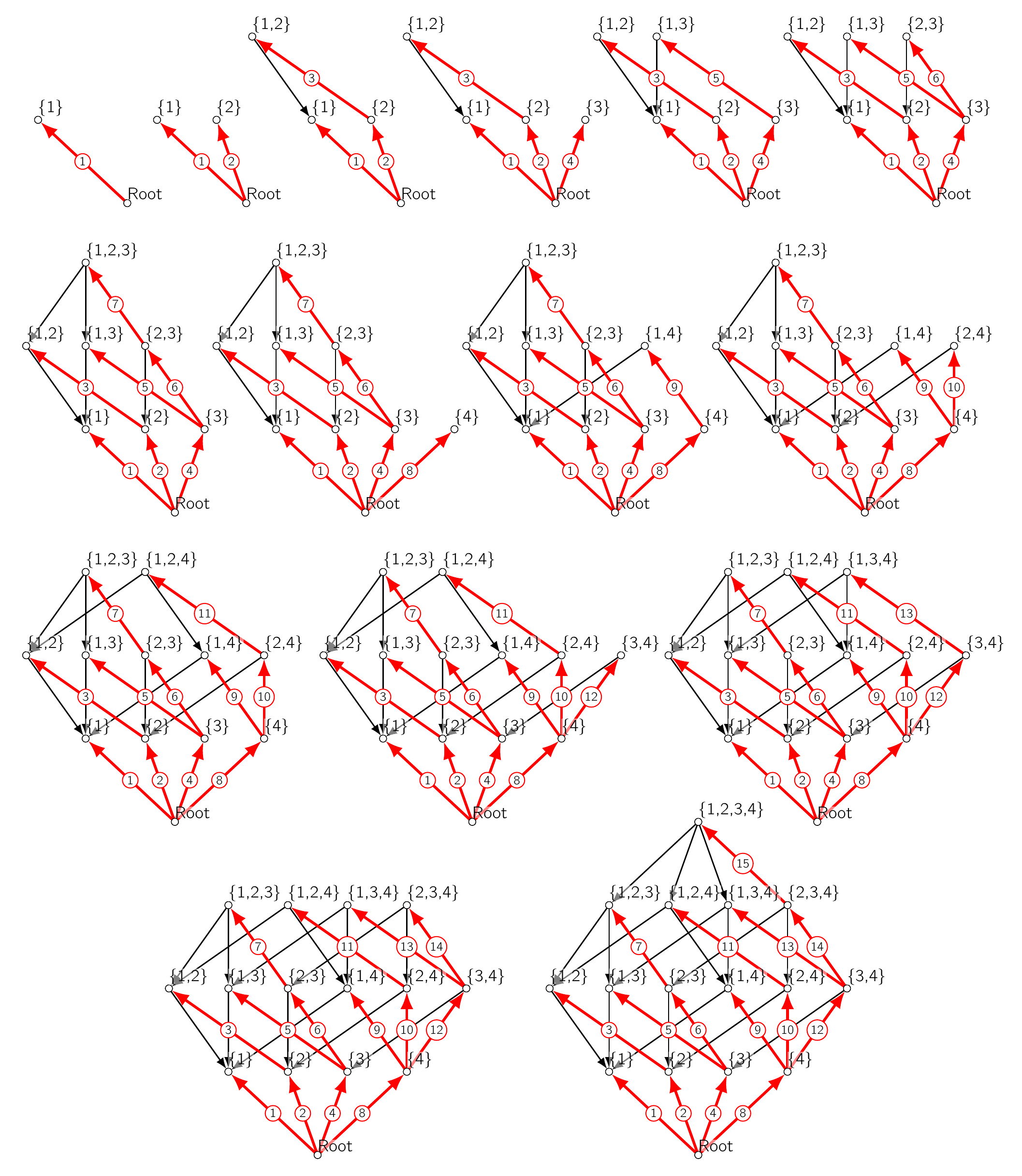}
  \caption{
    The Hasse diagrams for the step-by-step insertion of tetrahedron \{1,2,3,4\} by Algorithm~\ref{alg:insertion}. Red lines represent the order of creation for each simplex. The skinny black lines represent where connections to parent simplices are backfilled.
  }
  \label{fig:insertorder}
\end{figure}

\newpage
\subsection{Code for Getting Neighbors by Adjacency}\label{sec:actualneighbor}
\par The C++ code for collecting the set of \simplices{k} sharing a common coface with simplex $s$.
A function call to \verb|neighbors_up()| calls the following code which serves primarily to help the compiler deduce the dimension, $k$, of $s$.
\begin{lstlisting}
template <class Complex, class SimplexID, class InsertIter>
void neighbors_up(Complex &F, SimplexID s, InsertIter iter)
{
    neighbors_up<Complex, SimplexID::level, InsertIter>(F, s, iter);
}
\end{lstlisting}
With the simplex dimension determined, we call an overloaded function which defines the operation for a \simplex{k}.
\begin{lstlisting}
template <class Complex, std::size_t level, class InsertIter>
void neighbors_up(
        Complex &F,
        typename Complex::template SimplexID<level> s,
        InsertIter iter)
{
    for (auto a : F.get_cover(s))
    {
        auto id = F.get_simplex_up(s, a);
        for (auto b : F.get_name(id))
        {
            auto nbor = F.get_simplex_down(id, b);
            if (nbor != s)
            {
                *iter++ = nbor;
            }
        }
    }
}
\end{lstlisting}
Neighbors of $s$ are pushed into an insert iterator provided by the user.
In this fashion, depending upon the container type the iterator corresponds to, the user can specify whether or not duplicate simplices are returned (\verb|std::vector|) or not (\verb|std::set|).
\end{screenonly}
\end{document}